\documentclass[12pt]{article}
\usepackage{amsmath}
\usepackage{amsthm}
\usepackage{graphicx}
\usepackage{natbib}
\usepackage{url} 
\usepackage{amsmath,amsthm,amssymb,bbm}
\usepackage{mathtools}
\usepackage{mathrsfs}
\usepackage{hyperref}
\usepackage{subfigure}
\usepackage{preamble}
\usepackage{caption}
\newcommand{\blind}{0}

\begin{document}

\def\spacingset#1{\renewcommand{\baselinestretch}%
{#1}\small\normalsize} \spacingset{1}


\if0\blind
{
  \title{\bf Model-based Smoothing with Integrated Wiener Processes and Overlapping Splines}
  \author{Ziang Zhang \hspace{.2cm}\\
    Department of Statistical Sciences, University of Toronto\\
    and \\
    Alex Stringer \\
    Department of Statistics and Actuarial Science, University of Waterloo\\
    and \\
    Patrick Brown \\
    Department of Statistical Sciences, University of Toronto \\
    Centre for Global Health Research, St Michael's Hospital \\
    and \\
    Jamie Stafford \\
    Department of Statistical Sciences, University of Toronto \\
    }
  \maketitle
} \fi

\if1\blind
{
  \bigskip
  \bigskip
  \bigskip
  \begin{center}
    {\LARGE\bf Title}
\end{center}
  \medskip
} \fi

\bigskip
\begin{abstract}
In many applications that involve the inference of an unknown smooth function, the inference of its derivatives will often be just as important as that of the function itself. 
To make joint inferences of the function and its derivatives, a class of Gaussian processes called $p^{\text{th}}$ order Integrated Wiener's Process (IWP), is considered. 
Methods for constructing a finite element (FEM) approximation of an IWP exist but have focused only on the order $p = 2$ case which does not allow appropriate inference for derivatives, and their computational feasibility relies on additional approximation to the FEM itself.
In this article, we propose an alternative FEM approximation, called overlapping splines (O-spline), which pursues computational feasibility directly through the choice of test functions, and mirrors the construction of an IWP as the Ospline results from the multiple integrations of these same test functions.
The O-spline approximation applies for any order $p \in \mathbb{Z}^+$, is computationally efficient and provides consistent inference for all derivatives up to order $p-1$.
It is shown both theoretically, and empirically through simulation, that the O-spline approximation converges to the true IWP as the number of knots increases.
We further provide a unified and interpretable way to define priors for the smoothing parameter based on the notion of predictive standard deviation (PSD), which is invariant to the order $p$ and the placement of the knot.
Finally, we demonstrate the practical use of the O-spline approximation through simulation studies and an analysis of COVID death rates where the inference is carried on both the function and its derivatives where the latter has an important interpretation in terms of the course of the pandemic.
\end{abstract}

\noindent%

\spacingset{1.5} 
\section{Introduction}\label{sec:intro}

In many statistical applications that involve an unknown regression function, $g$, inference for the derivatives of $g$ are often as important as inference for $g$ itself \citep{li2020inference, swain2016inferring, de2015smoothed}. 
To make joint inference for $g$ and its derivatives, we consider a model-based smoothing approach that assigns a Gaussian process ($\GP$) model to the function and its derivatives \citep{rasmussen2003gaussian}. 

Due to the close relationship with the traditional smoothing spline \citep{wahbaimproper} the $p^{\text{th}}$ order Integrated Wiener Process, denoted as $\text{IWP}_p(\sigma)$, is a popular choice for the $\GP$ model \citep{rw208}.
The standard deviation parameter $\sigma \geq 0$ controls the covariance function of the $\text{IWP}_p(\sigma)$, and can also be interpreted as a smoothing parameter, where a larger value allows more variability in the inferred $g$ and a value close to zero will force the inferred $g$ to stay in the span of $p^{th}$ order polynomials.   

Given a $p-1$ times continuously differentiable function $g$ with $p \in \mathbb{Z}^+$, assigning $g$ with the $\text{IWP}_p(\sigma)$ immediately assigns its $q$th derivative $g^{(q)}$ with the $\text{IWP}_{p-q}(\sigma)$ for any $q < p$. 
Because of this {\em simultaneous derivative property}, model-based smoothing using $\text{IWP}_{p}(\sigma)$ yields interpretable joint inferences of $g$ and its derivatives. 
However, directly fitting the IWP models to $g$ and its derivatives is computationally intensive in many practical settings due to the cost to store and factorize the dense large covariance matrices. One way to significantly reduce the computational challenge is to approximate the IWP model with a finite-dimensional approximation obtained through the finite element method (FEM); see \cite{rw208} and \cite{adaptivesmoothingsplines} for examples. However, these methods only apply for order $p = 2$ and do not provide inference for derivatives. Furthermore, they involve an additional approximation to the FEM itself to achieve computational feasibility.
In this article, we propose an alternative FEM approximation, called overlapping splines (O-spline), which pursues computational feasibility directly through the choice of test functions, and mirrors the construction of an IWP as the O-spline results from the multiple integrations of these same test functions.
This article makes the following contributions:
\begin{enumerate}
    \item[(a)] We propose a computationally efficient, finite-dimensional approximation for the $\text{IWP}_p(\sigma)$ model through FEM, called overlapping splines (O-splines), which is suitable for any order $p\geq1$. We denote by $\tilde{g}_k$ the O-spline approximation for $g$, where $k$ is the number of knots used to construct the approximation.
    \item[(b)] We show both theoretically and through simulations that the joint distribution of $\tilde{g}_k$ and its derivatives converges to the distribution under the true $\text{IWP}_p(\sigma)$, as the number of knots $k$ increases.
    \item[(c)] We propose a unified way to define the priors for the parameter $\sigma$ based on the notion of $h$-units predictive standard deviation (PSD), which has consistent interpretation across different order $p$. 
\end{enumerate}

The rest of the paper is structured as the following. In \cref{sec:prelim}, we describe the modelling context for this paper and provide some necessary background for the $\text{IWP}_p(\sigma)$ model.
In \cref{sec:method}, we introduce the proposed O-spline approximation; discuss its statistical and computational properties that justify its usage, and describe how to efficiently fit the approximation with the computational method in \cite{noeps}. We also introduce a unified and interpretable way to define the prior for the parameter $\sigma$.
In \cref{sec:simulations} and \cref{sec:examples}, we illustrate the practical utility of the proposed method through simulation studies and an analysis of the COVID-19 death rates and their rates of change over time.
Finally, we conclude with a discussion in \cref{sec:discussion}.

The codes to replicate all the results and examples in this article can be found at the corresponding online repository \href{https://github.com/AgueroZZ/Smooth_IWP_code}{github.com/AgueroZZ/Smooth_IWP_code}.

\section{Smoothing with the Integrated Wiener Processes}\label{sec:prelim}

Consider the following hierarchical model:
\begin{equation}\label{equ:smoothModel}
    \begin{aligned}
    Y_i|\boldsymbol{\eta} &\overset{ind}{\sim} \pi(Y_i|\boldsymbol{\eta}, \kappa),  \ i \in [n],  \\
    \eta_i &= \boldsymbol{v}_i ^T \boldsymbol{\beta} + \sum_{l=1}^{L}g_{l}(x_{li}), \ x_{li} \in \Omega_l \subset \mathbb{R}, \boldsymbol{v}_i \in \mathbb{R}^r \\
    g_l:\ &\Omega \rightarrow \mathbb{R}, \  g_l \overset{ind}{\sim} \text{IWP}_{p_l}(\sigma_l), \forall l \in [L].
    \end{aligned}
\end{equation}
Here $\pi(Y_i|\boldsymbol{\eta}, \kappa)$ is a twice continuously-differentiable density, linear predictors $\boldsymbol{\eta} = \{\eta_i, i \in [n]\}$, covariates $\{\boldsymbol{v}_i,x_{li}: l \in [L], i \in [n]\}$ and hyperparameter $\kappa$. 
Each unknown function $g_l:\Omega_l \rightarrow \mathbb{R}$ is assigned an independent $\text{IWP}_{p_l}(\sigma_l)$ model with the SD parameter $\sigma_l$ and order $p_l$. Developing a technique for simultaneous inference of $g_l$ and its derivatives is the purpose of this paper.

The above class of model in \cref{equ:smoothModel} belongs to the extended latent gaussian models of \cite{noeps}, and it includes the commonly used generalized additive models \cite{hastiegambook} as well as their extensions such as \cite{casecrossover} and \cite{coxphus}. For the purpose of exposition in the next few sections we consider develops for a single unknown function $g$.

\subsection{Integrated Wiener Processes}\label{sec:IWP}

To permit simultaneous inference of $g$ and its derivatives up to the order $p-1$ we adopt a $p^{th}$ order Integrated Wiener Process ($\text{IWP}_p(\sigma)$) as the $\GP$ model for $g$ through the following construction:
\begin{equation}\label{eqn:expansion}
\begin{aligned}
g(x) \overset{d}{=} \sum_{l=0}^{p-1}\gamma_l x^l + \sigma W_p(x),
\end{aligned}
\end{equation}
which we denote as $g \sim$ $\text{IWP}_p(\sigma)$. 
Let $\boldsymbol{\gamma} = (\gamma_0, ..., \gamma_{p-1})^T \sim N(0, \Sigma_{\boldsymbol{\gamma}})$, where $\Sigma_{\boldsymbol{\gamma}} = \text{diag}(\tau_0^2, ..., \tau_{p-1}^2 )$. We assume by default that $\tau_l^2 = 1000$ for all $0\leq l<p$.
The process $W_p$ is independent of $\boldsymbol{\gamma}$, and is defined through
\citep{shepp1966radon}:
\begin{equation}\label{equ:iwp}
\begin{aligned}
\frac{\partial^q}{\partial t^q} W_p(0) = 0 ~~~\forall~ q<p \quad\quad\text{and}\quad\quad\frac{\partial^p}{\partial t^p} W_p(x) \overset{d}{=} \xi(x) 
 \ 
\end{aligned}
\end{equation}
where $\xi(x)$ is a generalized Gaussian white noise process \citep{harvey1990forecasting, spde}. 
The ability to conduct simultaneous inference for
$g$ and its derivatives is immediate from the above construction given $g^{(q)}\sim \text{IWP}_{p-q}(\sigma)$ with coefficients $\gamma_l,~q\leq l<p$. We call this the \emph{simultaneous derivative property}. Note throughout the order $p$ is treated as fixed and specified {\em a priori}. 

One barrier to conducting inference in this setting involves the covariance matrix of $W_{p}({\bf x})$ where $\boldsymbol{x} = \{x_i, i \in [n] \}$. While the covariance functions for $W_p$ and its derivative have explicit forms given in \citep{robinson2010continuous} the resulting covariance matrix is dense in each instance and inversion is computationally demanding requiring $O(n^3)$ floating point operations. A potentially efficient solution is to utilize the Markov property of integrated Weiner processes \citep{gmrfbook, robinson2010continuous} and make use of efficient algorithms for sparse matrix storage, decomposition, and inversion \citep{rue2001fast}. However, the size of the covariance matrix still grows with $n$, and computations can still become challenging with many locations. As a result, we propose a finite-dimensional approximation to $W_p$ hence $g$ using the Finite Element Method (FEM), that retains the \emph{simultaneous derivative property} while having desirable computational and theoretical properties.

\section{Finite Element Method and Overlapping Splines}\label{sec:method}

In this section, we develop a basis function approximation to $W_p$ (and hence $g$) using the Finite Element Method (FEM). Assume without loss of generality that the region of interest $\Omega$ has the form of an interval $[0,a]$ where $a \in \mathbb{R}^+$. Here we define $\widetilde{W}_p(x)$ and $\tilde{g}_k(x)$ as
\begin{equation}\label{equ:FEM-approxi}
    \begin{aligned}
    \widetilde{W}_p(x) = \sum_{j=1}^{k} w_j \varphi_j(x); \quad \tilde{g}_k(x) &= \sum_{l=0}^{p-1}\gamma_l x^l + \sigma \widetilde{W}_p(x)
    \end{aligned}
\end{equation}
where $\mathbb{B}_k := \{\varphi_i: i \in [k]\}$ is a set of basis functions to be chosen, and $\boldsymbol{w} = (w_1, ..., w_k)^T$ are unknown (random) weights to be inferred. 
The distribution of the weights is determined by further choosing a set of test functions $\mathbb{T}_k:=\{\phi_i, i \in [k]\}$, and enforcing the distributional approximation:
\begin{equation}\label{weakSol}
\boldsymbol{w}\sim N({\bf 0}, \Sigma_{\boldsymbol{w}})
\end{equation}
where $\Sigma_{\boldsymbol{w}}^{-1}={\bf B}^T{\bf T}^{-1}{\bf B}$ and $\bf B$, with ${\bf B}_{ij}=\int_\Omega \phi_i(x) \frac{\partial^p\varphi_j}{\partial t^p}(x) dx$ and ${\bf T}_{ij}=\int_\Omega \phi_i(x)\phi_j(x)dx$ respectively.
As such, the distribution of the FEM approximation $\widetilde{W}_p$ hence $\tilde{g}_k$ is completely determined by the choices of test functions $\mathbb{T}_k$ and basis function $\mathbb{B}_k$.
These choices also determine the other properties of the approximation. A more formal exposition of the above details, involving stochastic differential equations, is given in Supplement B but may also be found in the literature \citep{shepp1966radon}.

An example of the FEM is given in \cite{rw208} in what is formally referred to as a \emph{Galerkin solution} for an $\text{IWP}_2$ model. Here they choose $\mathbb{B}_k = \mathbb{T}_k$ to be linear B-splines. As a result, ${\bf T}^{-1}$ is dense and to compute the precision matrix efficiently the authors suggest a further approximation that involves replacing $\bf T$ with $\bf A$, a diagonal matrix where each entry is obtained as the sum of the corresponding row in $\bf T$.
The resulting approximation is called the \textit{continuous second order random walk} (RW2) and the same strategy is later generalized in \cite{spde} for modeling continuous spatial variation and in \cite{adaptivesmoothingsplines} for adaptive smoothing with $\text{IWP}_2$. While efficient, the RW2 method suffers from the following disadvantages:
\bigskip
\newline
\indent
a. The method is only defined for order $p=2$;
\newline
\indent
b. To be computationally feasible the method involves two approximations;
\newline
\indent
c. While the sample path from $\text{IWP}_2$ model is once continuously differentiable, the
\newline
\indent\indent
sample from the approximation is not.
As such the \emph{simultaneous derivative property}
\newline
\indent\indent
does not obtain.
\bigskip
\newline
In the next section, we address the above by proposing a new FEM approximation that pursues simultaneous inference and computational feasibility directly.


\subsection{Overlapping Splines (O-Splines)}\label{sec:Approxi}

In this section, we utilize the Finite Element Method (FEM) to derive a finite-dimensional approximation of $\text{IWP}_p(\sigma)$ model for general $p\in \mathbb{Z}^+$. We pursue the {\em simultaneous derivative property} and computational efficiency directly, develop what is formally referred to as a \emph{least squares solution}, and then demonstrate that the approximation has the desired properties. We call our proposed method an \emph{Overlapping-spline} (O-spline) approximation for reasons that are apparent from their construction. The
development is intuitive, to ensure $\bf T$ is diagonal we define the set of test functions $\mathbb{T}_k$ to be piece-wise constant
\begin{equation}\label{equ:basis}
 \phi_i(x) = 
 \begin{cases} 
 0 & \text{if } x < s_{i-1}, \\
 1   & \text{if } x \in (s_{i-1},s_{i}], \\
 0  & \text{if } x \geq s_{i},
 \end{cases}
 \end{equation}
where $\boldsymbol{s} = \{s_i, i \in [k]\} \subset \Omega$ is a set of $k$ unique locations in increasing order with $s_0=0$.
A least squares solution uses basis functions that satisfy the following conditions \citep{spde}:
\begin{equation}\label{equ:LeastSquare}
    \begin{aligned}
   \frac{\partial^p}{\partial x^p} \varphi_i(x) = \phi_i(x), \quad
   \frac{\partial^q}{\partial x^q} \varphi_i(0) = 0 \ \text{for} \  \forall ~~q<p.
    \end{aligned}
\end{equation}
which is assured by defining (O-Splines) basis $\mathbb{B}_k$ through repeated integration of the test functions $\mathbb{T}_k$:
\begin{equation}
 \varphi_i(x) = 
 \begin{cases} 
 0 & \text{if } x < s_{i-1}, \\
 \frac{(x-s_{i-1})^p}{p!}   & \text{if } x \in (s_{i-1},s_{i}], \\
 \sum_{k = 1}^{p} \frac{d_i^k (x-s_{i})^{p-k}}{k!(p-k)!}  & \text{if } x \geq  s_{i},
 \end{cases}
 \end{equation}
where $d_i = s_i - s_{i-1}$.
The above choice of basis and test function implies the two matrices $\bf T$ and $\bf B$ in \cref{equ:FEM-approxi} are identical. Hence the precision matrix of the basis weights in the proposed approximation is diagonal as $\Sigma_{\boldsymbol{w}}^{-1}={\bf B}^T{\bf T}^{-1}{\bf B} = {\bf T}$, with $i$th diagonal element being $\int_\Omega \phi^2_i(x)dx = d_i$.

From the above construction, it is immediately apparent that the O-spline approximation inherits the \emph{simultaneous derivative property}
\begin{equation}\label{equ:sim-property-OS}
    \begin{aligned}
        \frac{\partial^q}{\partial t^q}\widetilde{W}_p(x) = \frac{\partial^q}{\partial t^q} \sum_j^{k} w_j \varphi_j(x) = \sum_j^{k} w_j \varphi^{(q)}_j(x) = \widetilde{W}_{p-q}(x).
    \end{aligned}
\end{equation}

\noindent
In contrast to the RW2 method it also has the following three advantages:
\bigskip
\newline
\indent
a. The O-spline approximation handles IWP$_p(\sigma)$ at any choice of $p\geq 1$, and hence
\newline
\indent\indent
allows $p$ to be chosen based on the prior knowledge on the differentiability of $g$;
\newline
\indent
b. The sample path from the O-spline approximation inherits the same differentiability
\newline
\indent\indent
as the original IWP$_p(\sigma)$ model, which simultaneously yields model-based
\newline
\indent\indent
inference for the derivatives;    
\newline
\indent
c. Since the upper-trapezoidal and diagonal structures of $\Phi$ and $\Sigma_{\boldsymbol{w}}$ already facilitate
\newline
\indent\indent
efficient matrix computations, the O-spline approximation is directly derived as an
\newline
\indent\indent
FEM approximation to the IWP model without further matrix approximations.    
\bigskip
\newline
The O-spline approximation also has desirable theoretic properties described in the following section.

\subsection{Theoretic properties of the O-spline approximation}\label{sec:OSpline}

In addition to the immediate advantages of the proposed O-spline approximation it also has desirable theoretic properties. The proofs of these results appear in the appendix.

First, the proposed O-spline approximation avoids the ambiguity of knot selection and placement, which has been the central problem to address for spline-based approximation method \citep{eilers1996flexible}, such as discussed in \cite{eilers2010splines} and \cite{ruppert2000theory}. 
As stated in \cref{lem:convergence}, the covariance function of the proposed approximation will converge at a linear rate to the covariance of the true IWP$_p(\sigma)$ with any $p\geq 1$, as more knots get equally placed over the region of interest. 
As a result, in terms of approximation accuracy, it is better to use as many equally spaced knots as is computationally feasible.
\begin{lemma}[Convergence of O-spline Approximation]\label{lem:convergence} 
Let $\Omega = [a,b]$ where $a,b \in \mathbb{R}^+$ and let $g \sim \text{IWP}_p(\sigma)$ with $p \in \mathbb{Z}^+$. Assume the knots $\{s_1, ..., s_k\}$ are equally spaced over $\Omega$ for each $k \in \mathbb{N}$, and $\tilde{g}_k$ denotes the corresponding $p$th order (O-spline) approximation defined as in \cref{equ:FEM-approxi}, then:
$$||\C - \C_{k}||_{\infty} = O(1/k),$$
where $\C(s,t) = \Cov[g(s),g(t)]$, $\C_{k}(s,t) = \Cov[\tilde{g}_k(s),\tilde{g}_k(t)]$ and $||\C - \C_{k}||_{\infty} \equiv \sup_{s,t \in \Omega}|\C(s,t) - \C_{k}(s,t)|$.
\end{lemma}

Secondly, the \emph{simulateous derivative property} of the O-spline approximation implies the knots sequence and the weight coefficients used for the function will be the same as those for the derivatives and the only additional step is to recompute the design matrix at a lower order.
This property makes our proposed O-spline approximation appropriate and convenient for the joint inference of the function with its derivatives, and the following \cref{thrm:convergence-joint} can be proved from this property and \cref{lem:convergence}:
\begin{theorem}[Main Theorem]\label{thrm:convergence-joint} 
Given the same setting and notations as in \cref{lem:convergence},  for any non-negative integers $q_1 \leq p-1$ and $q_2 \leq p-1$:
$$||\C^{(q_1,q_2)} - \C^{(q_1,q_2)}_{k}||_{\infty} = O(1/k),$$
where $\C^{(q_1,q_2)}(s,t) = \Cov[g^{(q_1)}(s), g^{(q_2)}(t)]$ and $\C^{(q_1,q_2)}_{k}(s,t) = \Cov[\tilde{g}^{(q_1)}_k(s), \tilde{g}^{(q_2)}_k(t)]$.
\end{theorem}
This theorem gives general convergence results for the proposed O-spline approximation, implying that every cross-covariance and hence the cross-correlation between the approximation and its derivatives will converge to the true value.
In particular, that implies any finite-dimensional distributions (f.d.d) of $\tilde{g}_k$ and its derivatives will converge to the true f.d.d under the $\text{IWP}_p(\sigma)$ model:
\begin{corollary}[Convergence of Finite Dimensional Distribution]\label{cor:fdd} 
Consider the same setting and notations as in \cref{lem:convergence}, let $p,m \in \mathbb{Z}^+$, then:
$$[\tilde{g}^{(q_1)}_k(x_1), \tilde{g}^{(q_2)}_k(x_2), ..., \tilde{g}^{(q_m)}_k(x_m)]^T \overset{d}{\rightarrow} [g^{(q_1)}(x_1), g^{(q_2)}(x_2), ..., g^{(q_m)}(x_m)]^T,$$
as $k \rightarrow \infty$, where $\{x_i \in \Omega:i \in [m]\}$ and $\{q_i \in \mathbb{Z}^+: i \in [m], q_i < p\}$ are arbitrary.
\end{corollary}
\begin{proof}
    Using \cref{thrm:convergence-joint} and Cramer-Wold device, \cref{cor:fdd} directly follows from Levy's continuity theorem.
\end{proof}
\noindent
Note this result can not be achieved for the RW2 approximation method because its sample path is not $p-1$ times continuously differentiable.

\subsection{Prior and interpretation of $\sigma$}\label{sec:intepret_var}

A key consideration in our context is the choice of an appropriate prior for the parameter $\sigma$. This is made complex given the interpretability of $\sigma$ depends on the order $p$. For example, IWP models with the same value of $\sigma$, but different orders, can have sample paths with extremely different variability making the choice of prior problematic for practitioners.

To address this we adopt a strategy similar to \cite{scalingigmrf} and assign a prior to a quantity that still involves $\sigma$ but has a consistent interpretation across different values of $p$. Aiming for interpretability we consider the standard deviation of the function $g$ at a future location $x+h$ given the value of $g$ and its derivatives at location $x$, denoted as:
\begin{equation}\label{equ:condition_var}
    \begin{aligned}
    \sigma(h) = \text{Var}\bigg[g(x+h) \bigg| g(x), g^{(1)}(x), ... , g^{(p-1)}(x)\bigg]^{1/2} = \frac{\sqrt{h^{(2p-1)}} \sigma}{\sqrt{(2p-1)}(p-1)!}.
    \end{aligned}
\end{equation}
We refer this quantity $\sigma(h)$ as the $h$-units predictive SD (PSD), which quantifies the uncertainty in predicting the function $g$ at $h$ units ahead, using its information up to the current location.
The choice of the unit $h$ can be made based on its practical relevance.
This differs from the approach of \cite{scalingigmrf}, which uses the marginal standard deviation, but benefits from $\sigma(h)$ being invariant to the location $x$ while the marginal standard deviation is not. At the same time, \cite{scalingigmrf} applies their prior elicitation method on the finite-dimensional approximation of the $\text{IWP}_p(\sigma)$ process, whereas our approach is entirely based on the original $\text{IWP}_p(\sigma)$ process, hence invariant to the approximation choices such as the number and placement of the knots.

Finally, motivated by \cite{pcprior}, we consider an exponential prior for $\sigma(h)$ of the form $\text{P}(\sigma(h) > u) = \alpha$, where $\{u,\alpha\}$ will be chosen by the user. 
The prior for $\sigma$ is then recovered by scaling the exponential prior for $\sigma(h)$ with $(p-1)!\sqrt{(2p-1)/h^{(2p-1)}}$.

\subsection{Approximate Bayesian Inference}\label{sec:comp}

For the remainder of this paper we embed the above developments in the ELGM context of \cite{noeps}. Using the adaptive Gaussian Hermite quadrature (\texttt{aghq}) algorithm introduced there, we outline how the proposed O-spline approximation in \cref{equ:FEM-approxi} may be implemented within an ELGM in \cref{equ:smoothModel}.

Following the details of \cite{noeps} we let $\regparam = (\boldsymbol{w}, \boldsymbol{\gamma},\boldsymbol{\beta}),~ \varparam = (\sigma, \kappa)$ and have $\regparam\sim\text{N}\left[\boldsymbol{0},\Sigma(\varparam)\right]$ where $\Sigma(\varparam) = \text{diag}(\sigma^2\Sigma_{\boldsymbol{w}}, \Sigma_{\boldsymbol{\gamma}}, \Sigma_{\boldsymbol{\beta}})$ with $\Sigma_{\boldsymbol{\beta}}$ and $\Sigma_{\boldsymbol{\gamma}}$ being diagonal matrices. The goal is to make fully Bayesian inferences for $\regparam,\varparam$ based on the posterior distributions $\pi(\regparam|\response)$ and $\pi(\varparam|\response)$.
Since both posterior distributions are intractable, we compute their corresponding approximations $\laplaceapprox(\varparam|\response)$ and $\widetilde{\pi}(\regparam|\response)$ using the methods of \cite{noeps}. First we compute the following Gaussian and Laplace approximations:
\begin{equation}\label{equ:Gaussian_Laplace}
\begin{aligned}
\gaussapprox(\regparam|\varparam,\response) = N(\wmode{\varparam}, \whess{\varparam}^{-1}), \ 
    \laplaceapprox(\varparam,\response) = \frac{\pi(\wmode{\varparam},\varparam,\response)}{\gaussapprox(\wmode{\varparam}|\varparam,\response)},
\end{aligned}
\end{equation}
where $\wmode{\varparam} = \argmax\log\pi(\regparam,\varparam,\response)$ is the mode and $\whess{\varparam}$ is the negative Hessian evaluated at the mode.
Explicitly, we write $\boldsymbol{\eta}$ as $\boldsymbol{\eta} = \boldsymbol{\Phi} \boldsymbol{w} + \boldsymbol{P} \boldsymbol{\gamma} + 
\boldsymbol{V} \boldsymbol{\beta}$, where $\boldsymbol{V}, \boldsymbol{P}$ and $\boldsymbol{\Phi}$ are the appropriate design matrices, and compute $\whess{\varparam}$ as:
\begin{equation}
\begin{aligned}
\whess{\varparam} &=
\begin{pmatrix}
\sigma^2\Sigma_{\boldsymbol{w}} & 0 & 0 \\ 0 & \Sigma_{\boldsymbol{\gamma}} & 0 \\ 0 & 0 & \Sigma_{\boldsymbol{\beta}}
\end{pmatrix} + 
\begin{pmatrix}
\boldsymbol{\Phi}^T \\ \boldsymbol{P}^T \\ \boldsymbol{V}^T
\end{pmatrix}
~{\partial_{\boldsymbol{\eta}}^2\ \log \pi(\response|\hat{\boldsymbol{\eta}}(\varparam), \kappa)}
~\begin{pmatrix}\boldsymbol{\Phi}\ \boldsymbol{P}\ \boldsymbol{V}
\end{pmatrix}
\end{aligned}
\end{equation}
We then utilize the \texttt{aghq} package in the \texttt{R} language \citep{aghqsoftware,aghqus} to compute the normalized Laplace approximation $\laplaceapprox(\varparam|\response)$ and hence $\widetilde{\pi}(\regparam|\response)$.  
With a sample $\{\regparam_m\}_{m=1}^{M}$ from $\widetilde{\pi}(\regparam|\response)$ available, a sample $\{\widetilde{g}_{m}(t)\}_{m=1}^{M}$ or its derivatives can be obtained, for any $t\in \Omega$.


\section{Simulation Study}\label{sec:simulations}

\subsection{Assessment of approximation accuracy}\label{sec:sim1}

In this section, we will assess the quality of the O-spline approximation to the IWP model.
To compare the results across different choices of $p$, we will compute the correlation functions instead of the covariance functions, respectively using the true IWPs and their O-spline approximations.
Without the loss of generality, we assume $\gamma_i = 0$ for each $0 \leq i < p$ and $\sigma = 1$ in \cref{eqn:expansion}.

For order $p = 1,2,3$ and $4$, we compute the auto-correlation function $\rho(5,x) = \text{Cor}[W_p(5),W_p(x)]$ of the true IWP. We then obtain the corresponding approximations using O-spline with respectively $k = 5, 10, 30$ or $100$ knots placed equally over the interval $[0,15]$. As shown in Figure \ref{fig:converg}(a-d), the approximations using O-spline are always accurate for all $p$ when $k \geq 30$. As $p$ increases, the improved approximation quality with higher $k$ becomes more obvious at larger values of $x$.

To assess the quality of approximation in the sense of joint distributions, we then compare the cross-correlation function between the IWP with its $q$th derivative, denoted as $\rho^{(0,q)}(5,x) = \text{Cor}[W_p(5), W^{(q)}_p(x)]$.
We consider $p = 2,3,4$ and $q = 1$ or $2$. Their O-spline approximations are obtained using the same setting as above. The results are summarized in Figure \ref{fig:converg}(e-h). Similar to the previous result on the auto-correlation, the cross-correlations are always accurately approximated by the O-spline approximation for all $p$ when $k \geq 30$.

The two results above together suggest that our proposed O-spline approximations are practically accurate at different orders of $p$ even with small $k$, both in terms of the distributions within the IWP as well as the joint distribution with its derivatives.

\begin{figure}[!p]
    \centering
             \subfigure[IWP $p=1$]{
      \includegraphics[width=0.23\textwidth]{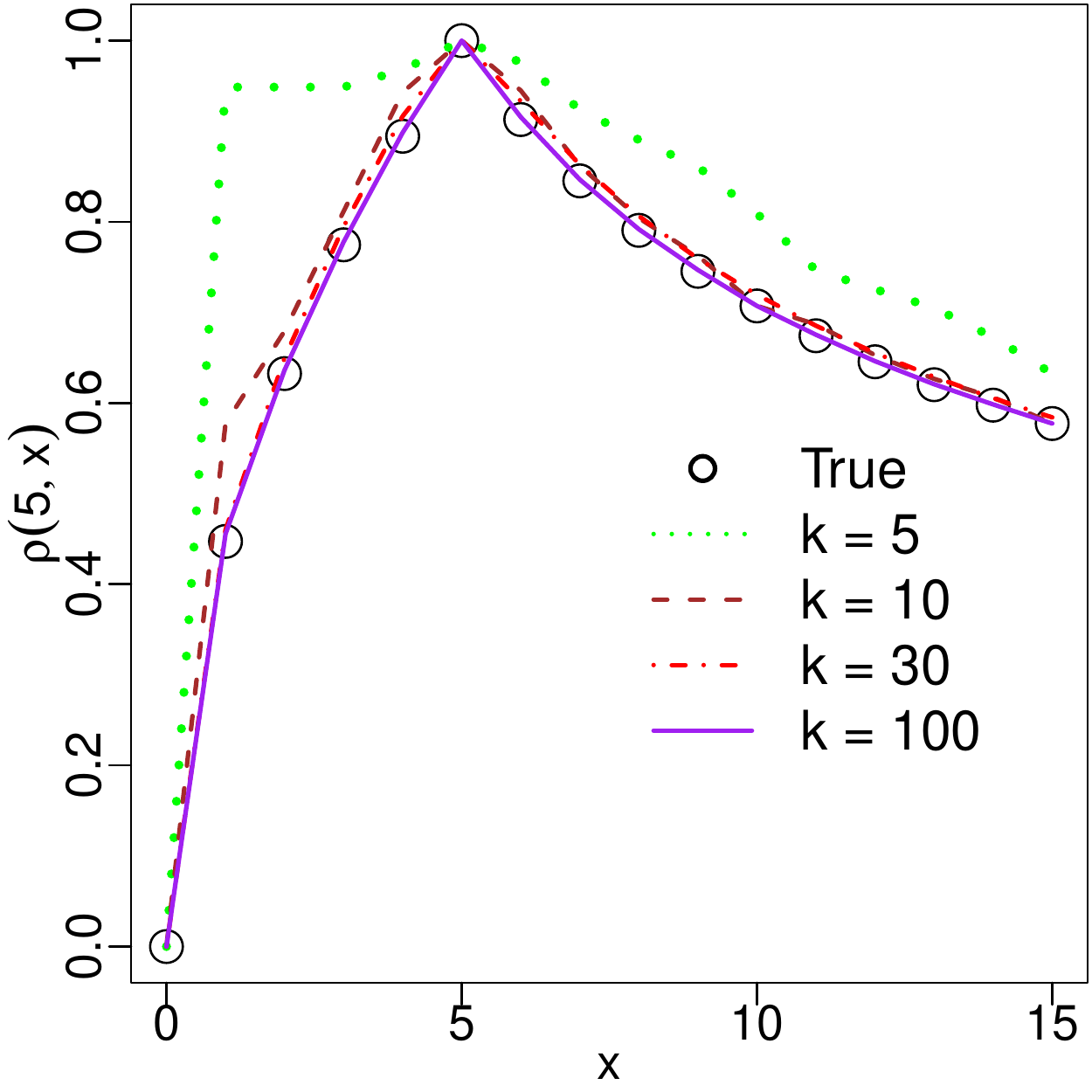}
      }\vspace{-0.5\baselineskip}%
      \subfigure[IWP $p=2$]{
            \includegraphics[width=0.23\textwidth]{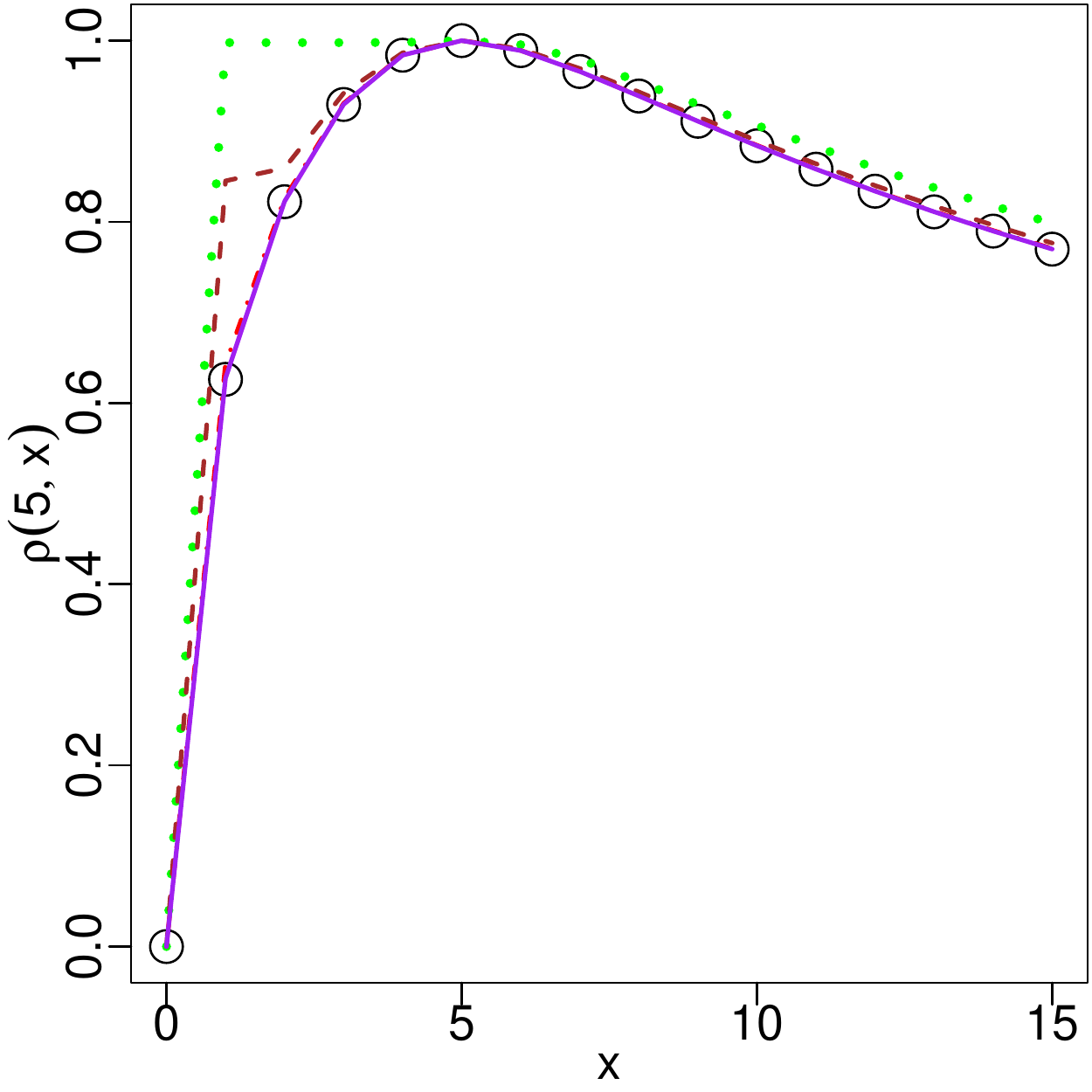}
            }%
            \subfigure[IWP $p=3$]{
                  \includegraphics[width=0.23\textwidth]{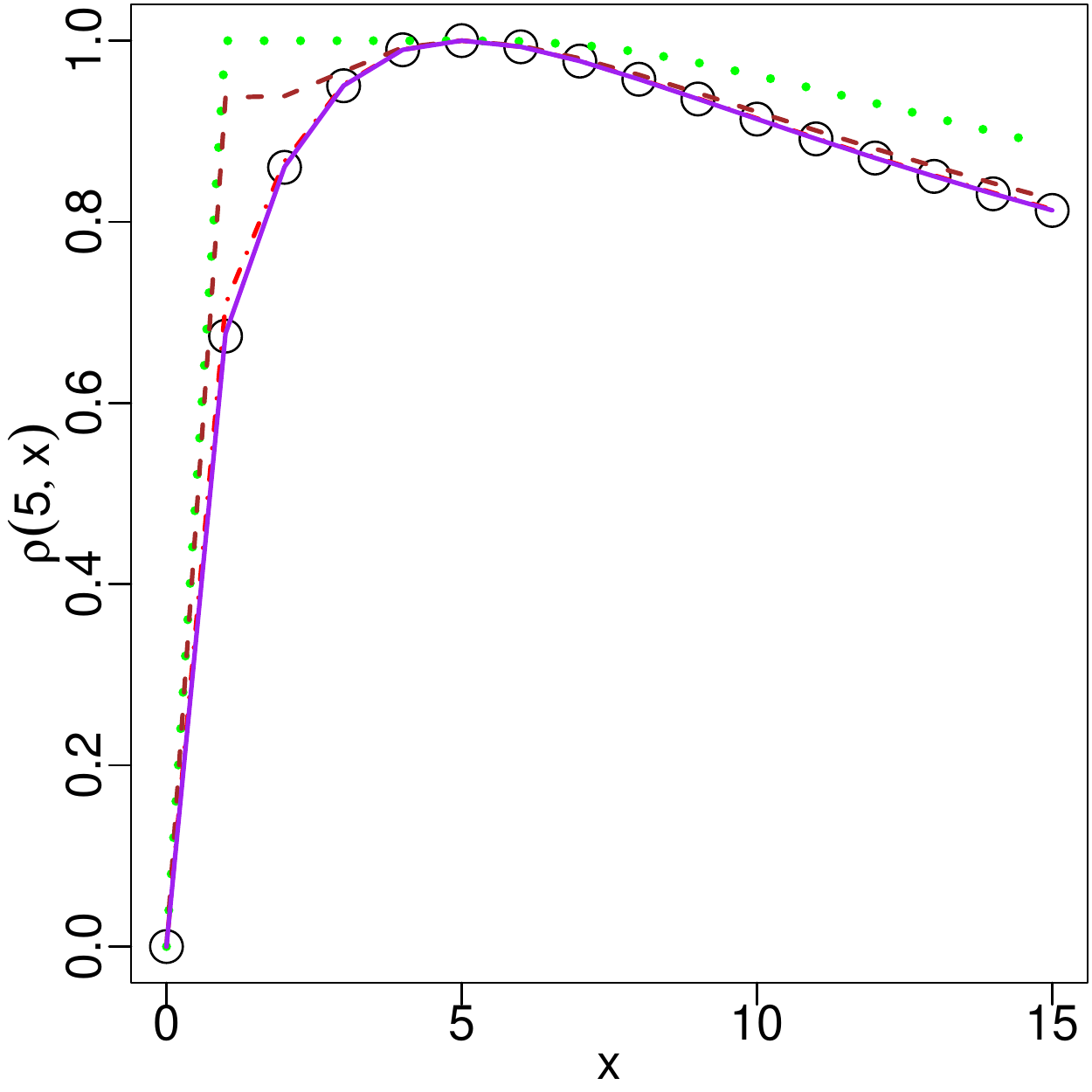}
                  }%
                  \subfigure[IWP $p=4$]{
                  \includegraphics[width=0.23\textwidth]{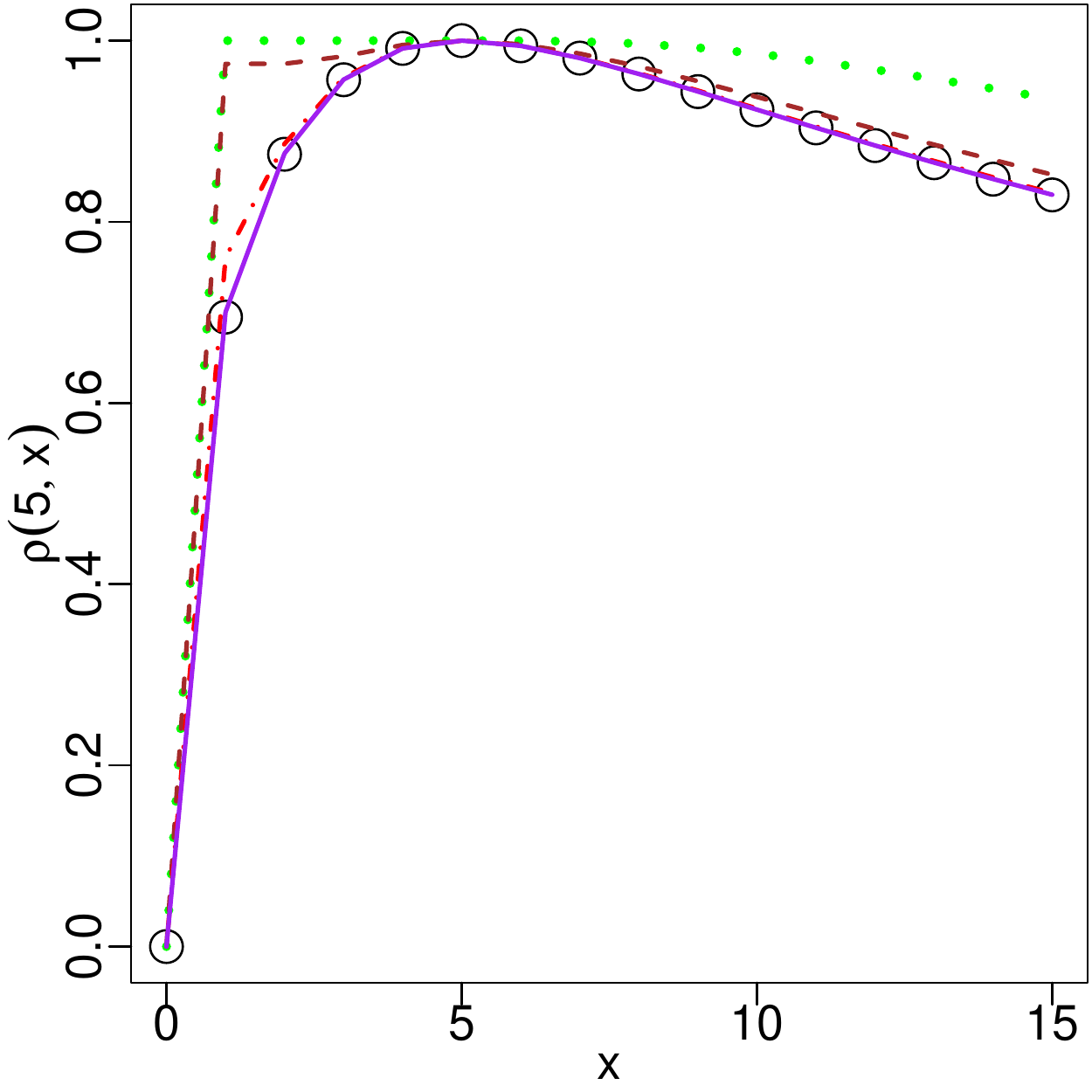}
                  }
                 \subfigure[IWP $p=2$]{
      \includegraphics[width=0.23\textwidth]{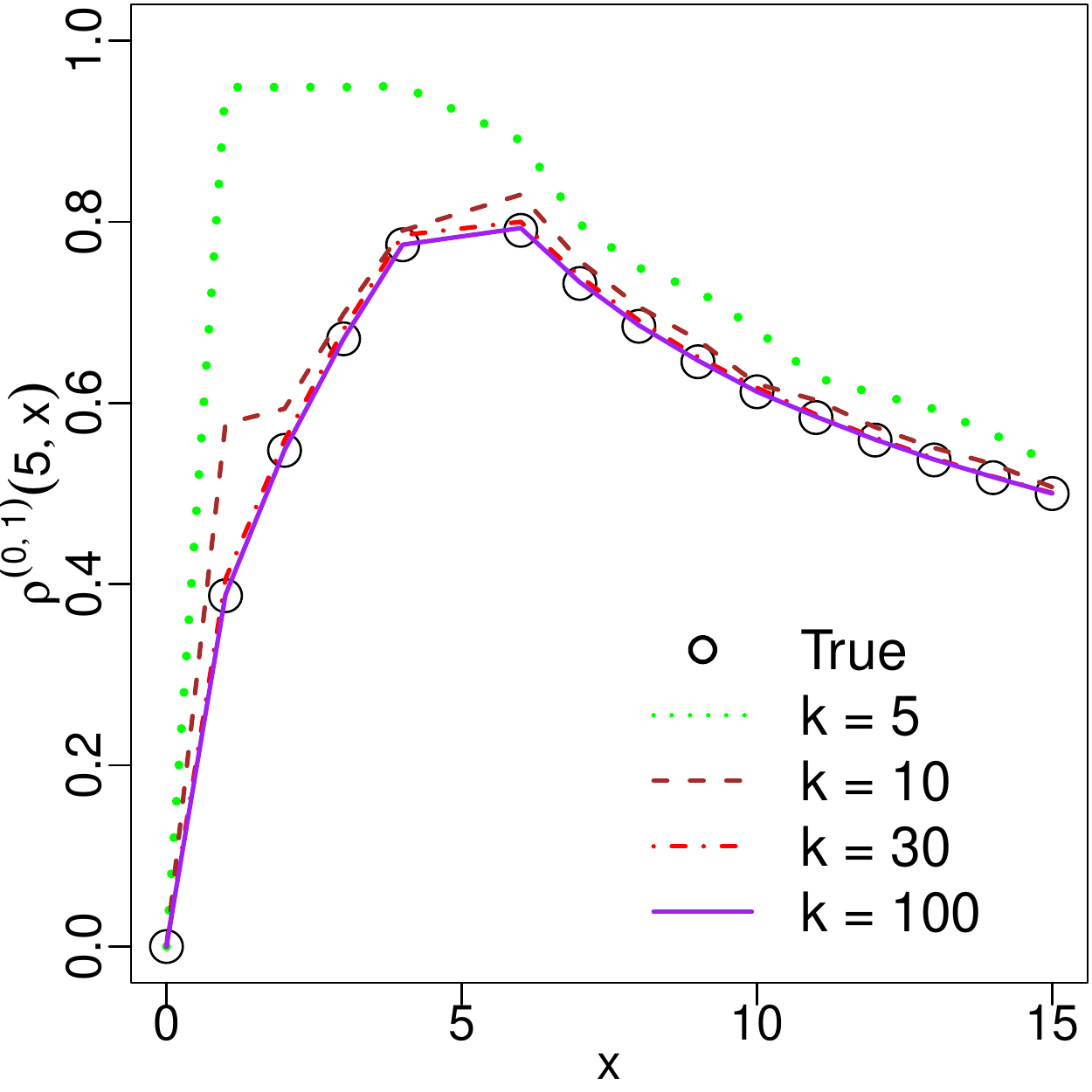}
      }%
      \subfigure[IWP $p=3$]{
      \includegraphics[width=0.23\textwidth]{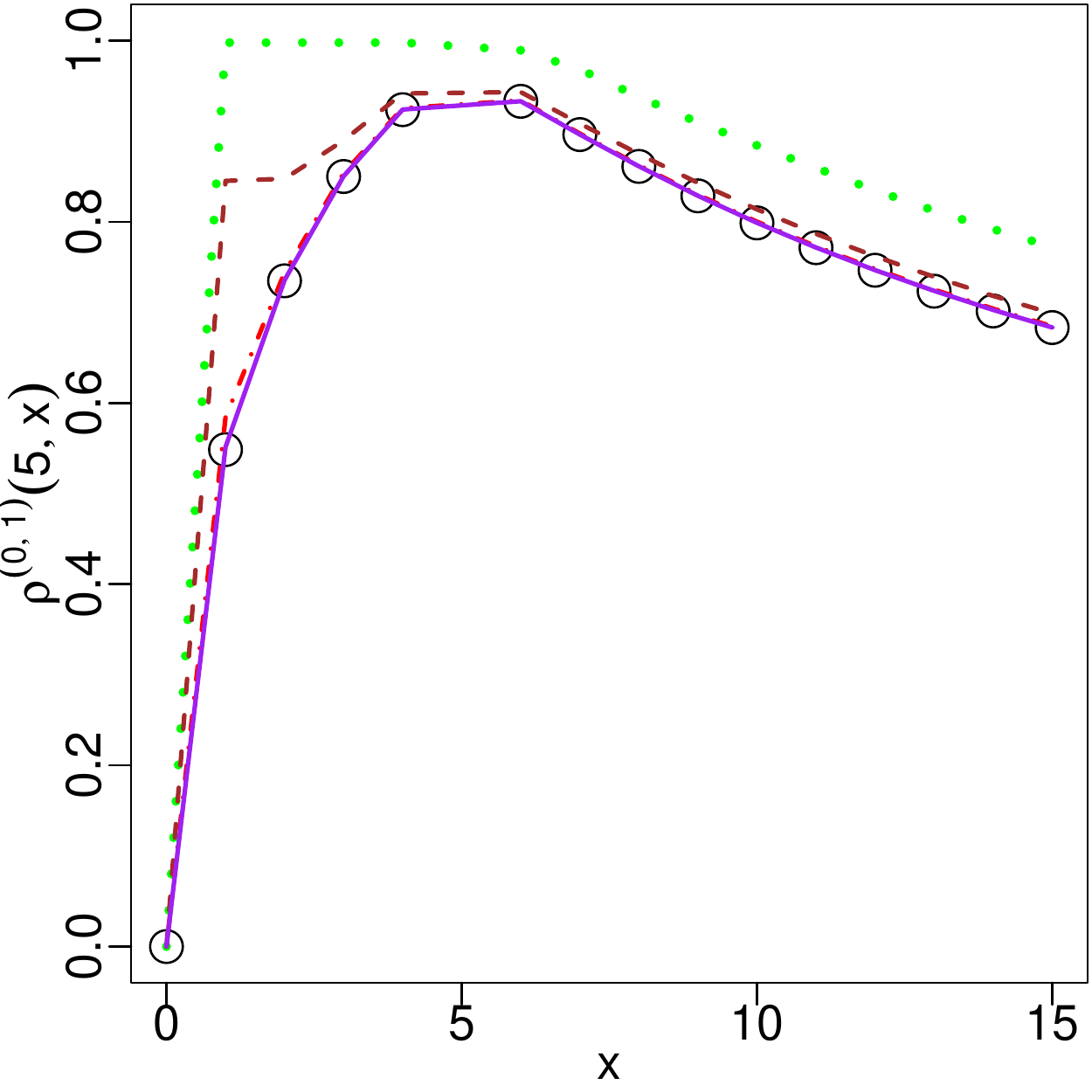}
      }%
      \subfigure[IWP $p=4$]{
      \includegraphics[width=0.23\textwidth]{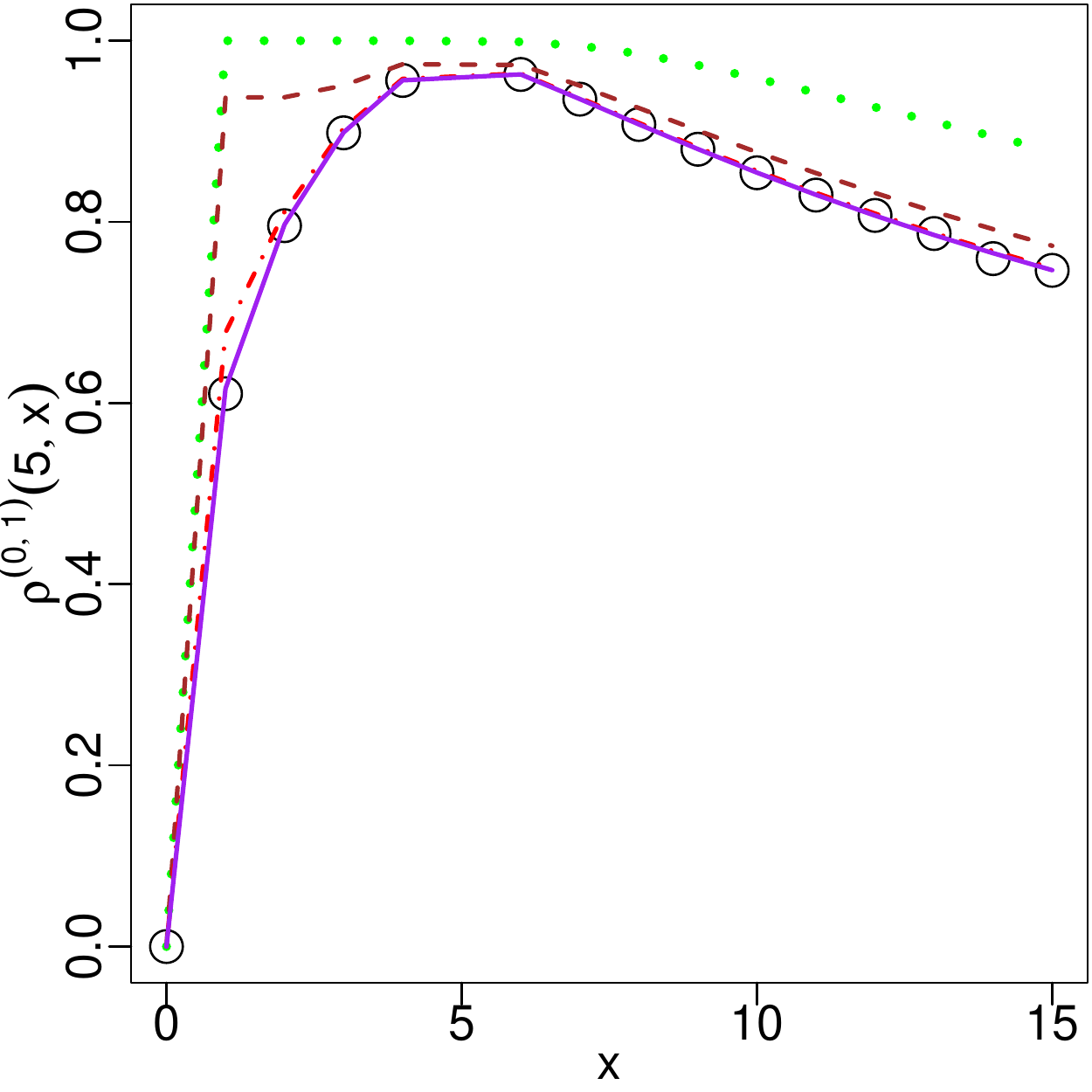}
      }%
      \subfigure[IWP $p=4$]{
      \includegraphics[width=0.23\textwidth]{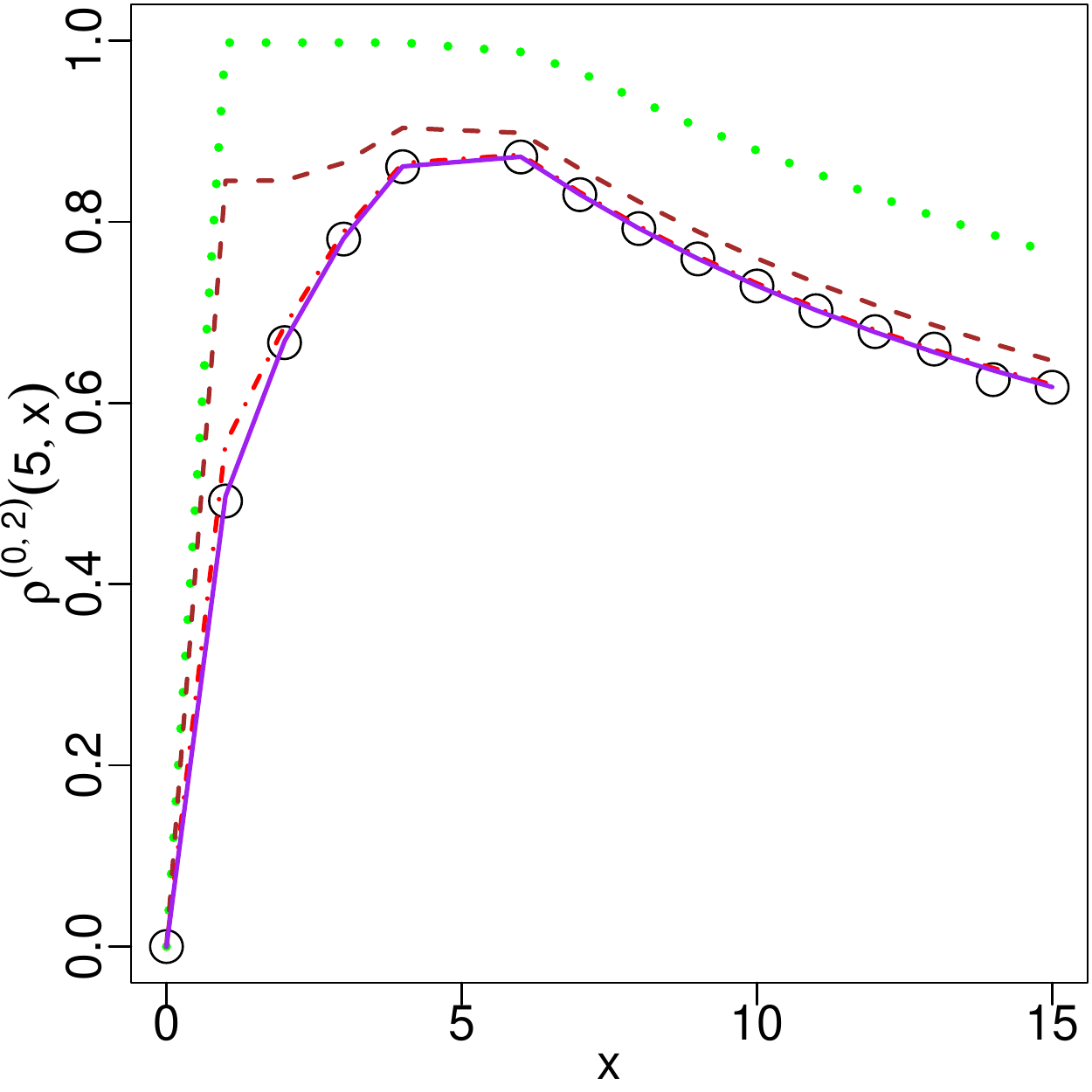}
      }

     \caption{Auto-correlation $\rho$ (a-d) and cross-correlation $\rho^{(0,q)}$ (e-h) functions of the IWP with its O spline approximation at $p = 1,2,3$ and $4$ for Section \ref{sec:sim1}. The plots (e-g) compute the correlation between the IWP with its first derivative, and the last plot in (h) computes the correlation between IWP with its second derivative.
     The auto-correlation and cross-correlation are accurately approximated when $k \geq 30$.}
    \label{fig:converg}
\end{figure}

\subsection{Computational Comparison with exact IWP}\label{sec:compareOSIWP}

In this section, we will illustrate the computational feasibility of the proposed O-spline approximation by comparing its implementation with that of the exact IWP model using the approach of augmented space as described in \cite{gmrfbook} and \cite{robinson2010continuous}. 
To ensure the comparability of the results, both approaches will be fitted using the approximation Bayesian inference method as described in \cref{sec:comp}.

As a proof of concept, we assume $\boldsymbol{x} = \{x_i\}_{i=1}^n$ were equally placed over $\Omega = [0,20]$, and $p=3$. The $n$ observations were simulated from the following simple univariate regression model:
\begin{equation}\label{equ:sim42mod}
\begin{aligned}
    y_i &= g(x_i) + \epsilon_i; \quad  x_i \in \Omega; \\
    \epsilon_i &\sim N(0,1) ;\quad g(x) = \sqrt{3}\sin(x/2).
\end{aligned}
\end{equation}

For the inference of $g$, we consider both the exact $\text{IWP}_3(\sigma)$ and its O-spline approximation.
The five-unit PSD $\sigma(5)$ was given Exponential prior with $\text{P}(\sigma(5)>3) = 0.01$.
The number of knots $k$ used in the O-spline approximations are respectively $10,30,50$ and $100$, equally placed over $\Omega$.
All computations in \cref{sec:comp} were done single-threaded with the number of adaptive quadrature being $10$ and the number of posterior samples being $3,000$.

To compare the numerical stability of each method, we compute the condition number 
of $\whess{\varparam}$ in \cref{sec:comp} defined as
\begin{equation}\label{equ:maxCN}
    \begin{aligned}
    \kappa[\whess{\varparam}] = \frac{\lambda_{\text{max}}[\whess{\varparam}]}{\lambda_{\text{min}}[\whess{\varparam}]},
    \end{aligned}
\end{equation}
where $\lambda_{\text{max}}$ and $\lambda_{\text{min}}$ denote the largest and smallest singular value respectively.
We then compute the maximum condition number $\kappa_{max}$ of each method defined as 
\begin{equation}\label{equ:maxCN}
    \begin{aligned}
    \kappa_{\text{max}} = \text{max} \{\kappa[\whess{\varparam_j}]: j \in [10]\},
    \end{aligned}
\end{equation}
where $\{\varparam_j : j \in [10]\}$ denotes the set of 10 quadrature points.
Table \ref{table:Simtable} displays the maximum condition number of each method for different $n$. Note that when $n \geq 800$, the implementation exact IWP method fails due to serious numerical singularity problems. However, the proposed O-spline approximations has no numerical problems.

We also compared the average runtime for each method.
For each choice of $n$, each model will be fitted independently for $10$ times, and the relative runtime is compared with the average runtime of the O-spline approximation with $k=10$ knots and $50$ observations.
Table \ref{table:Simtable} shows the mean and standard deviation of the 10 relative runtimes for each model. The proposed O-spline approximation works much faster than the implementation using the exact IWP method, even when a large of number of knots ($k \geq n$) is used in the approximation. 


Finally, to illustrate the approximation quality using the proposed O-spline approach, we compared the inferential results obtained from O-spline approximations and the exact IWP method when $n=100$, for both the function and its first two derivatives.
As shown in the Figure \ref{fig:simCompareIWPOS}, the inferential results between the exact IWP method and its O-spline approximation are very similar in posterior mean for both the function $g$ and its derivatives, even with only $k = 10$. The posterior standard deviations for the higher order derivative have some degree of inconsistency when $k$ is small, but it gets much smaller as $k$ increases to around $30$.

In summary, the proposed method is found to yield indistinguishable inferential results compared to the exact IWP method, but with a significantly shorter runtime and better numerical stability when the number of locations of interest is large.

\begin{table}
  \begin{center}
  \scalebox{0.58}{
  \begin{tabular}
  {|p{1cm}|p{2.5cm}|p{1.90cm}|p{2.5cm}|p{1.90cm}|p{2.5cm}|p{1.90cm}|p{2.5cm}|p{1.90cm}|p{2.5cm}|p{1.90cm}|}
   \hline
     \cline{1-11}
        \multicolumn{1}{c}{} &
        \multicolumn{2}{c}{Exact} &
        \multicolumn{2}{c}{$\text{OS}_{k=10}$} &
        \multicolumn{2}{c}{$\text{OS}_{k=30}$} &
        \multicolumn{2}{c}{$\text{OS}_{k=50}$} &
        \multicolumn{2}{c}{$\text{OS}_{k=100}$} \\
        \hline
        $n$ &  {Rel Runtime} & {CN $(\log_{10})$} & {Rel Runtime} & {CN $(\log_{10})$} & {Rel Runtime} & {CN $(\log_{10})$} & {Rel Runtime} & {CN $(\log_{10})$} & {Rel Runtime} & {CN $(\log_{10})$}\\
           \hline
         50 & 3.29(0.06) & 12.52 & 1.00(0.05) & 6.35 & 1.10(0.06) & 6.10 & 1.31(0.33) & 6.03 & 1.65(0.30) & 5.98 \\
         100 & 5.47(0.33) & 14.13 & 1.11(0.03) & 6.47 & 1.23(0.01) & 6.20 & 1.37(0.05) & 6.12 & 1.86(0.04) & 6.10 \\
         200 & 10.40(0.62) & 16.05 & 1.83(0.38) & 6.98 & 1.83(0.05) & 6.76 & 2.02(0.03) & 6.96 & 2.82(0.06) & 7.23 \\
         500 & 41.44(1.21) & 17.34 & 5.80(0.34) & 7.19 & 6.08(0.21) & 7.08 & 6.69(0.35) & 7.30 & 8.02(0.24) & 7.56 \\
         800 & --(--) & -- & 14.53(0.17) & 7.29 & 15.07(0.28) & 7.24 & 15.76(0.08) & 7.46 & 18.15(0.24) & 7.73 \\
         2000 & --(--) & -- & 111.25(1.68) & 7.49 & 112.00(1.16) & 7.68 & 113.47(0.90) & 7.86 & 120.68(1.38) & 8.20 \\
         5000 & --(--) & -- & 1000.72(25.03) & 7.81 & 996.54(8.13) & 7.90 & 999.43(5.04) & 8.09 & 1010.83(6.77) & 8.43 \\
  \hline
  \end{tabular}
  }
  \end{center}
  \caption{The left column (Rel Runtime) shows the mean (std.dev) of relative runtimes for 10 runs of the exact IWP model and of each of the O-spline approximation, as described in \cref{sec:compareOSIWP}. 
  The relative runtimes are computed by the dividing the average runtimes of O-spline with $k=10, n = 50$.
  The right column (CN) shows the $\kappa_{\text{max}}$ at $\log_{10}$ scale.
  For $n\geq 800$, the exact method fails due to numerical problem, whereas the O-spline method is unaffected. The O-spline method has both faster runtime and better numerical stability regardless of $k$, and the difference gets bigger as $n$ increases.
  }
  \label{table:Simtable}
\end{table}

\begin{figure}[!p]
    \centering
                   \subfigure[$\text{SD}\lbrack g(x)|y \rbrack$]{
      \includegraphics[width=0.3\textwidth]{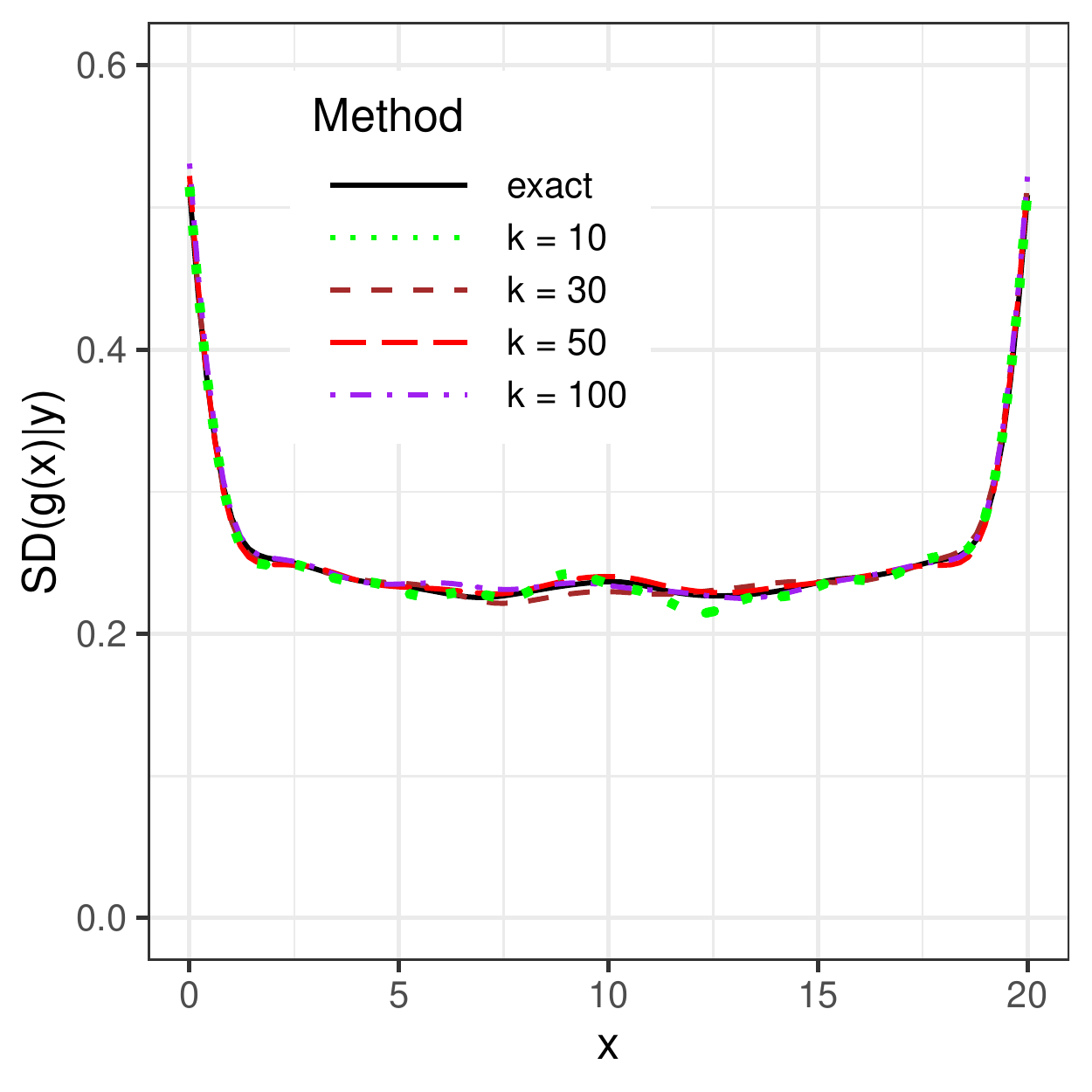}
    }
              \subfigure[$\text{SD}\lbrack g'(x)|y \rbrack$]{
      \includegraphics[width=0.3\textwidth]{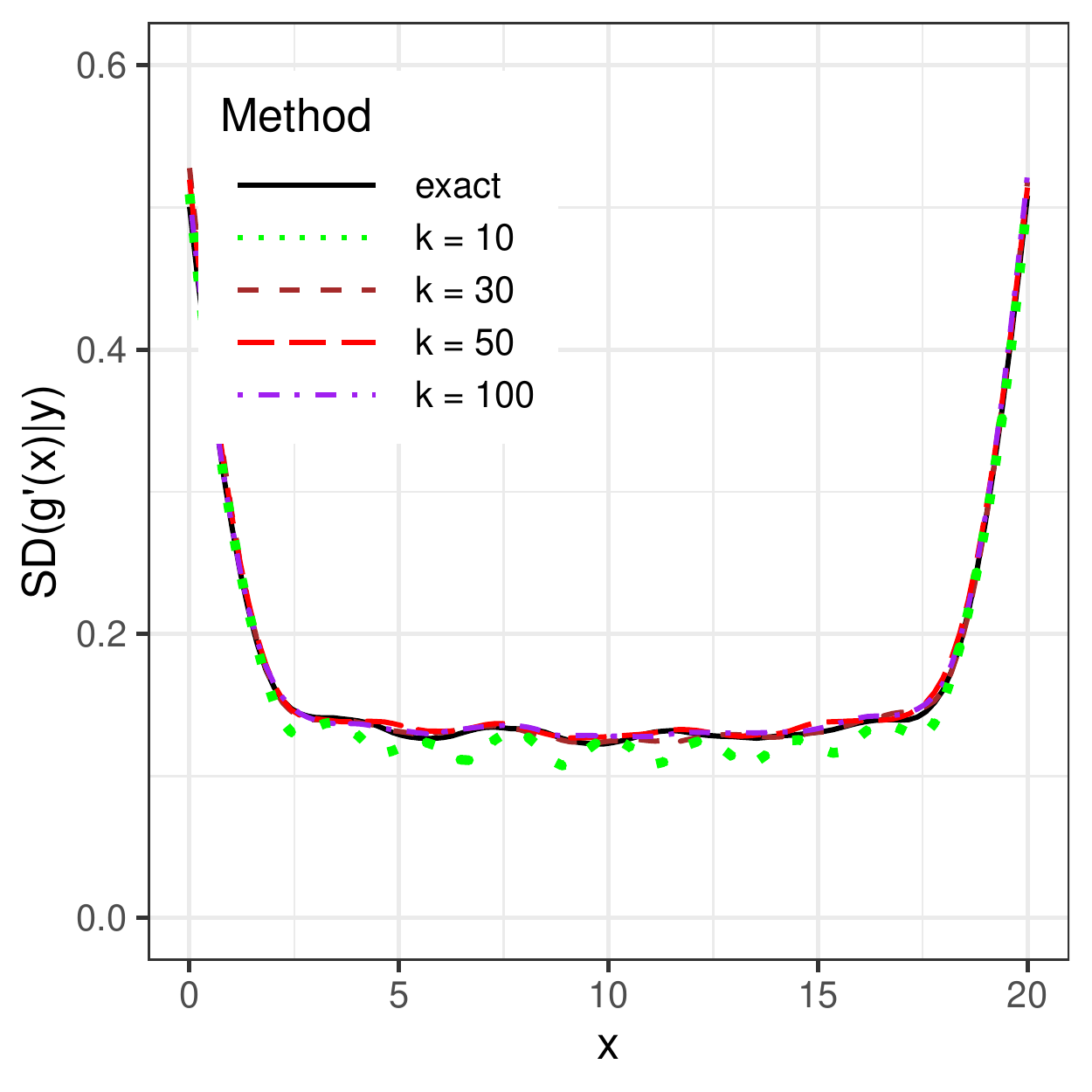}
    }
               \subfigure[$\text{SD} \lbrack g''(x)|y \rbrack$]{
      \includegraphics[width=0.3\textwidth]{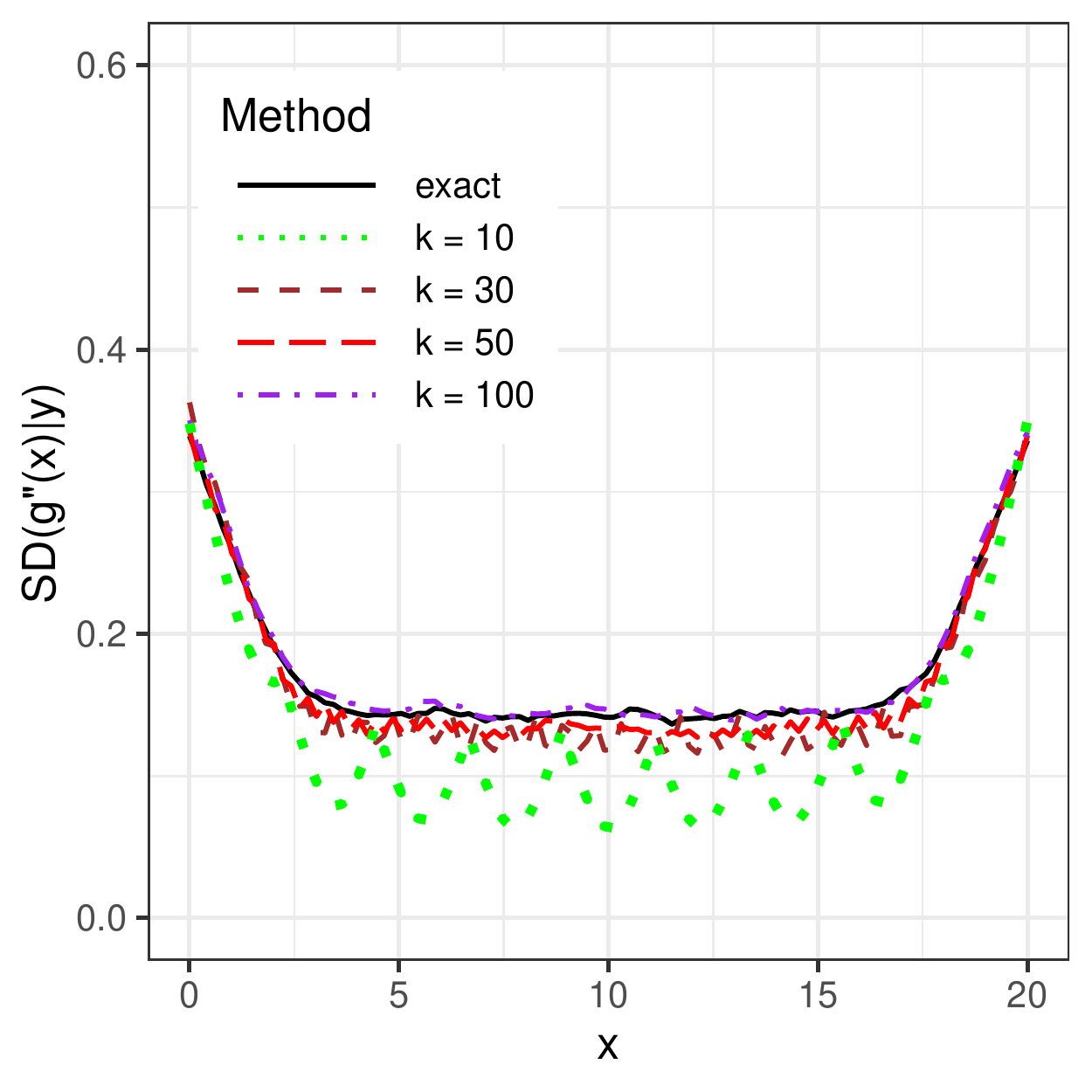}
    }
             \subfigure[$\E \lbrack g(x)|y \rbrack $]{
      \includegraphics[width=0.3\textwidth]{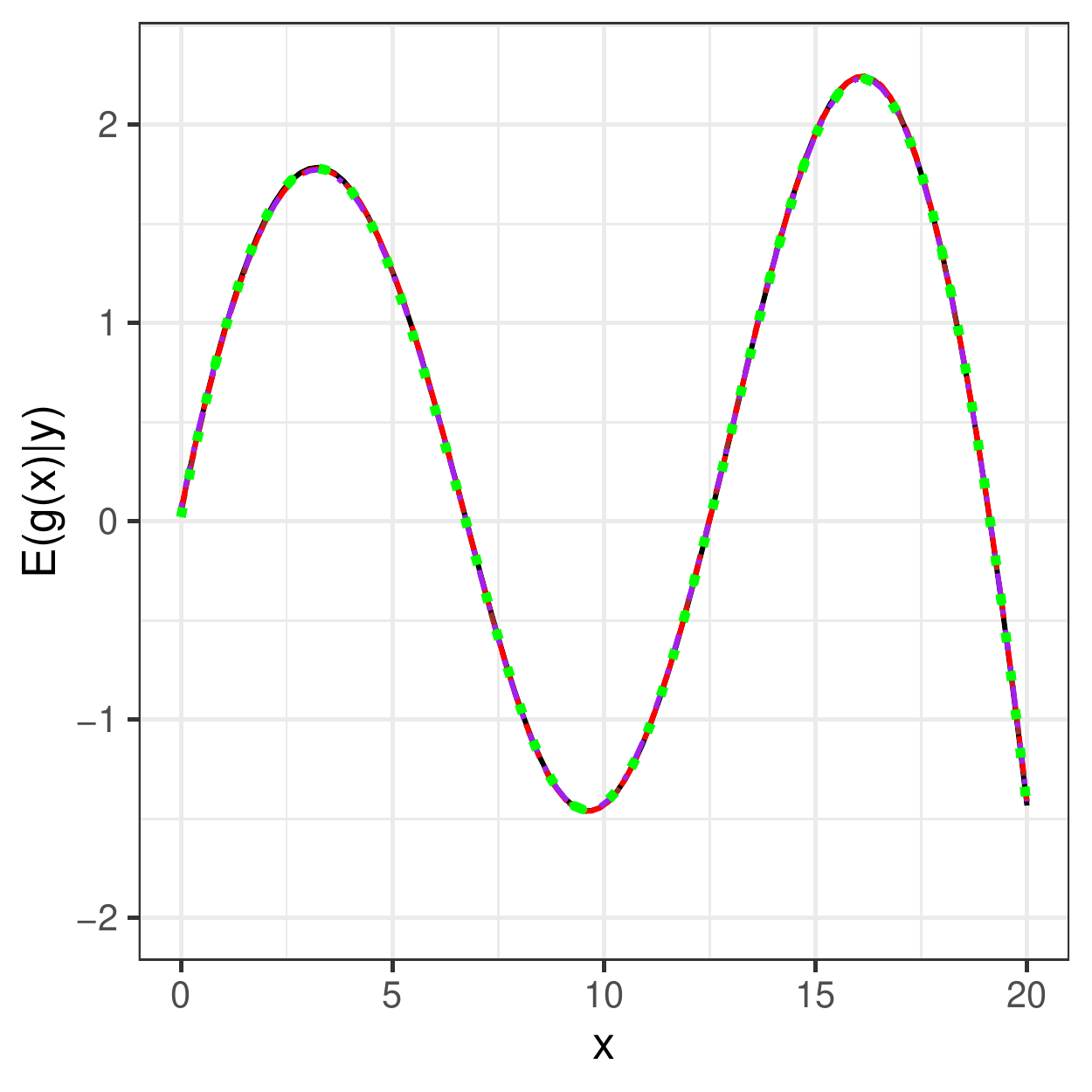}
    }
              \subfigure[$\E \lbrack g'(x)|y \rbrack $]{
      \includegraphics[width=0.3\textwidth]{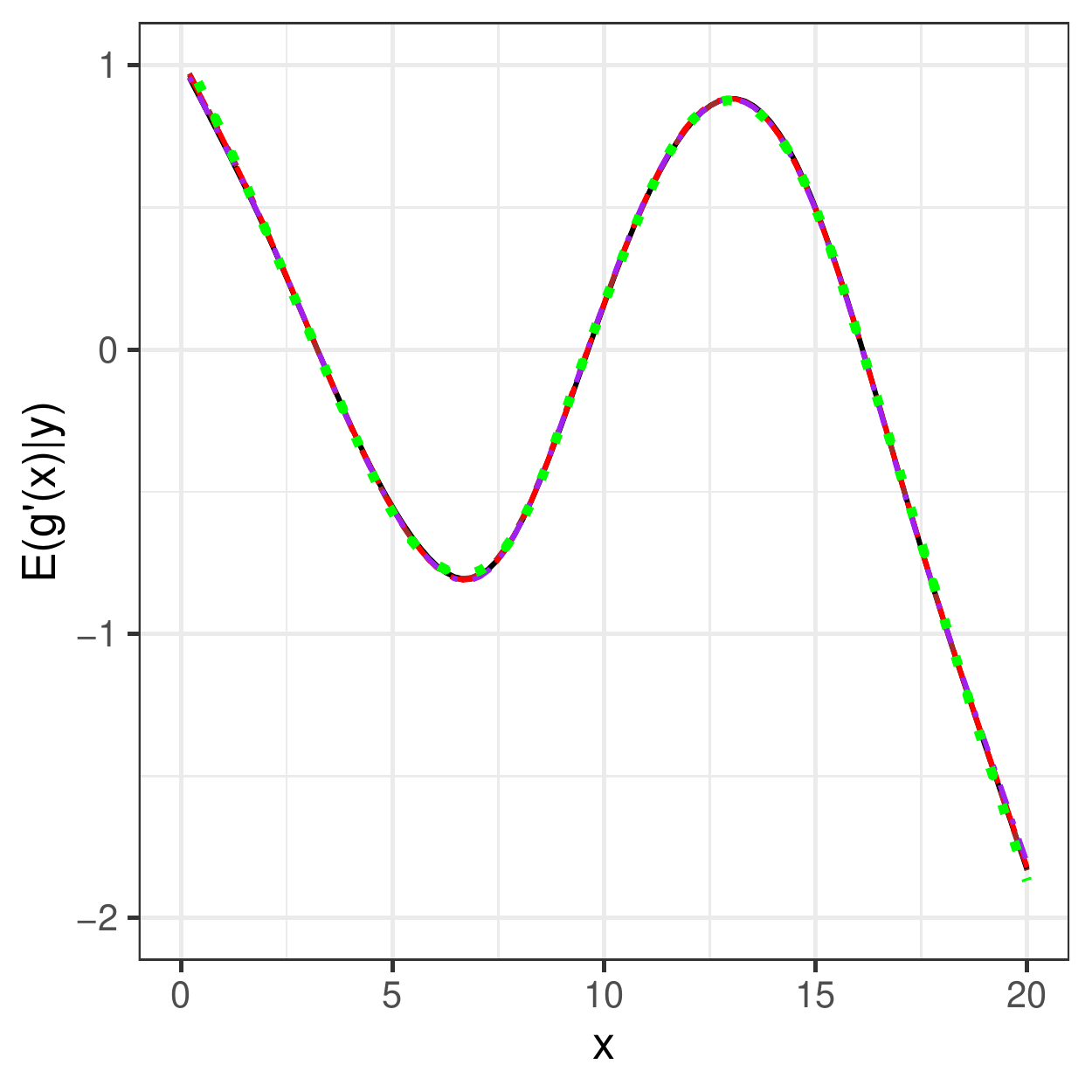}
    }
               \subfigure[$\E \lbrack g''(x)|y \rbrack $]{
      \includegraphics[width=0.3\textwidth]{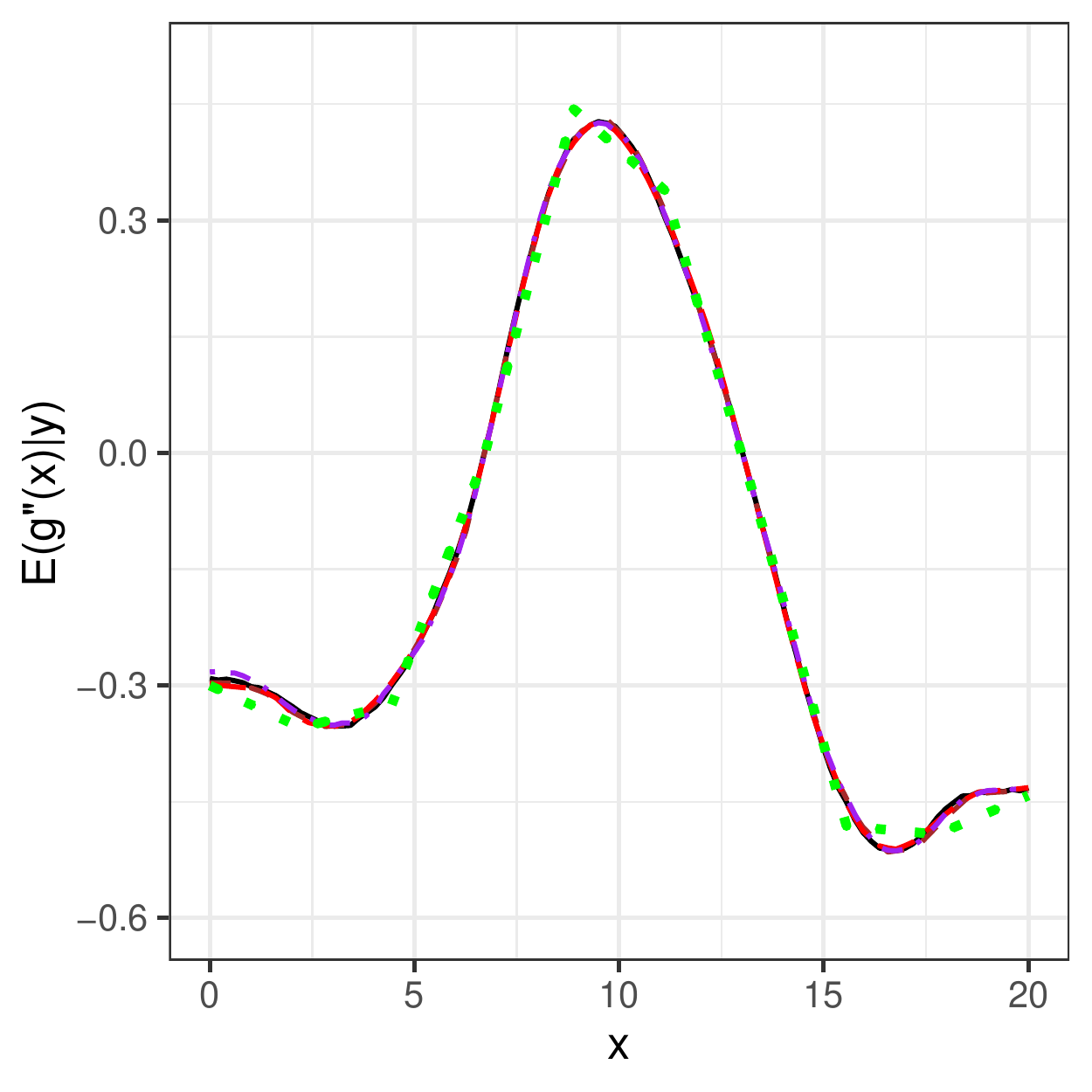}
    }
     \caption{Comparing the posterior mean and standard deviation from the exact IWP method and its O-spline approximations, as described in \cref{sec:compareOSIWP} when $n = 100$, for both the function and its derivatives. The O-spline approximation is very accurate for both the posterior mean and the standard deviation when $k \geq 30$.}
    \label{fig:simCompareIWPOS}
\end{figure}

\subsection{Assessment of inferential accuracy}\label{sec:sim2}

In this section, we illustrate that the higher-order O-spline method can provide inferential results that are significantly improved compared with the existing methods based on second-order smoothness, especially for derivatives.

We consider the true function $g(x) = \sum_{i=1}^{3} \delta_i \phi(x-\mu_i)$ simulated with the form of mixture Gaussian density, where $\phi$ denotes the standard Gaussian density. The mixture weights $\{\delta_1, \delta_2, \delta_3 \}$ are set to $\{0.6, 0.3, 0.1\}$ and the mean values $\{\mu_i; i \in [3] \}$ are independently simulated from $N(5, 4)$ in each replication. The observation locations $\boldsymbol{x} = \{x_1, ..., x_n\}$ are equally located over $[0,10]$ with $n = 100$. The simulated true function will then be standardized so that the sample variance of $g(\boldsymbol{x})$ equals to one in each replication.
Each observation is simulated as $y_i = g(x_i) + e_i$, where $e_i \overset{iid}{\sim} N(0, 1/100)$, making a signal-to-noise ratio of 100. 

For comparison, we fit the continuous RW2 method by \cite{rw208} which is derived as an approximation to IWP$_2(\sigma)$, and the proposed O-spline approximation to IWP$_3(\sigma)$. All the methods are fitted with 100 equally placed knots, and implemented using the approximate Bayesian inference method in \cref{sec:comp}.
We adopt the approach in \cref{sec:intepret_var} with an Exponential prior on the one-unit PSD $\sigma(1)$ such that $\text{P}(\sigma(1) > 1) = 0.5$, which is then scaled to the corresponding Exponential prior for $\sigma$ for $p=3$ in the O-spline implementation and for $p=2$ in the RW2 implementation. 


The rMSE between the posterior mean functions and the true function is computed as $\text{rMSE}(g) = \sqrt{{\sum_{i=1}^n \{\mathbb{E}[g(x_i)|\boldsymbol{y}] - g(x_i)\}^2}/{n}}$, and similarly for the first and second-order derivatives. 
To compare the inferential accuracy between different methods, we then scale the rMSE of each method by the median rMSE of the O-spline method.
We repeat the above procedure for $300$ independent replications. 


The results are summarised in Figure \ref{fig:sim_GMM_agg}. It shows that for such smooth true functions with a few peaks and valleys in their derivatives, the O-spline method with third-order smoothness indeed performs better than the existing method based on IWP-$2$. The difference is most obvious for the inference of the second derivative $g''(\boldsymbol{x})$ where the second-order method yields overly wiggly estimate. The main reason behind this is that the proposed O-spline method with third-order smoothness assumes the sample path from the target process to be twice continuously differentiable, whereas the second-order method RW2 assumes the sample path from the target process to be only once continuously differentiable.

\begin{figure}[!p]
    \centering
             \subfigure[rMSE ratio for $g$]{
      \includegraphics[width=0.31\textwidth]{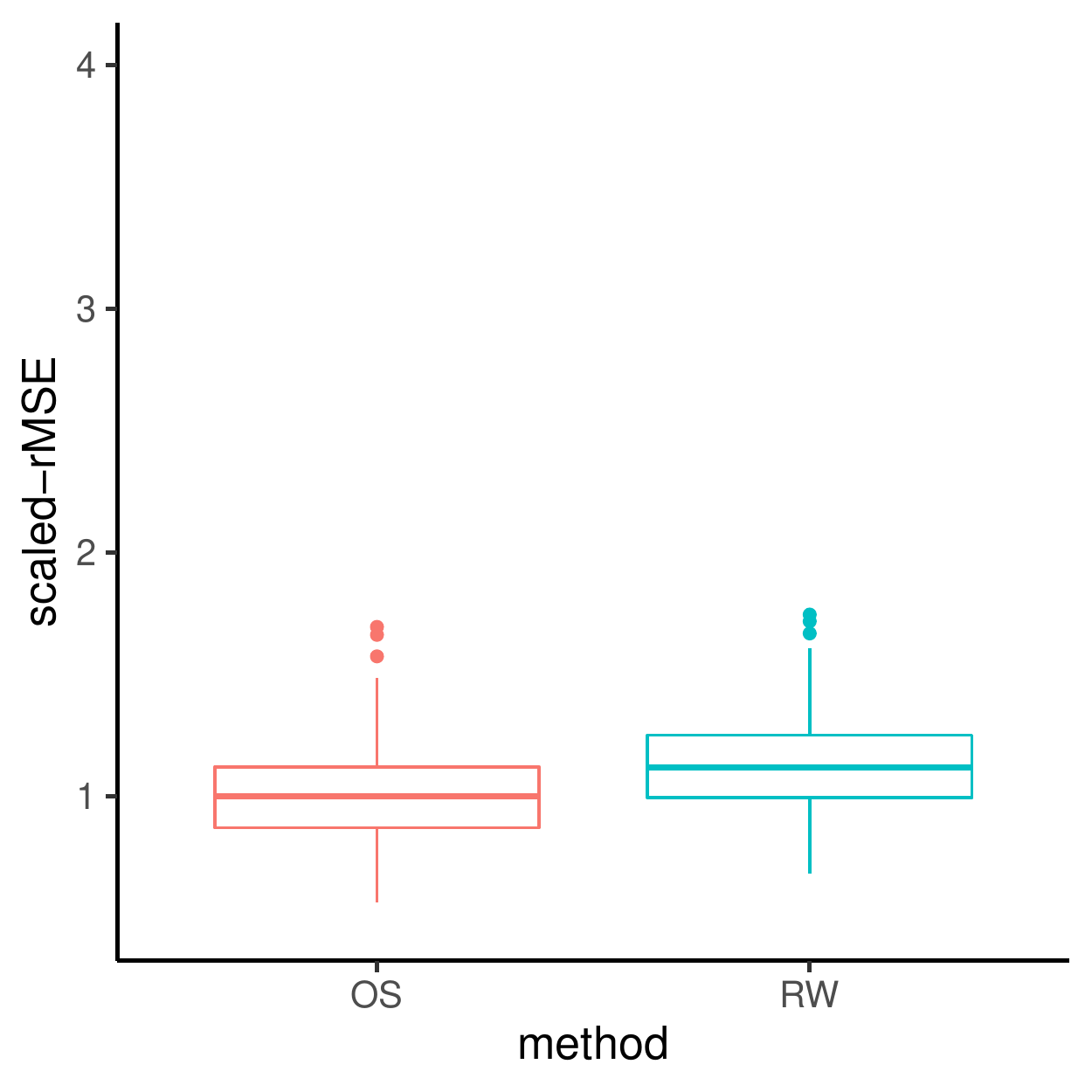}
    }
             \subfigure[rMSE ratio for $g'$]{
      \includegraphics[width=0.31\textwidth]{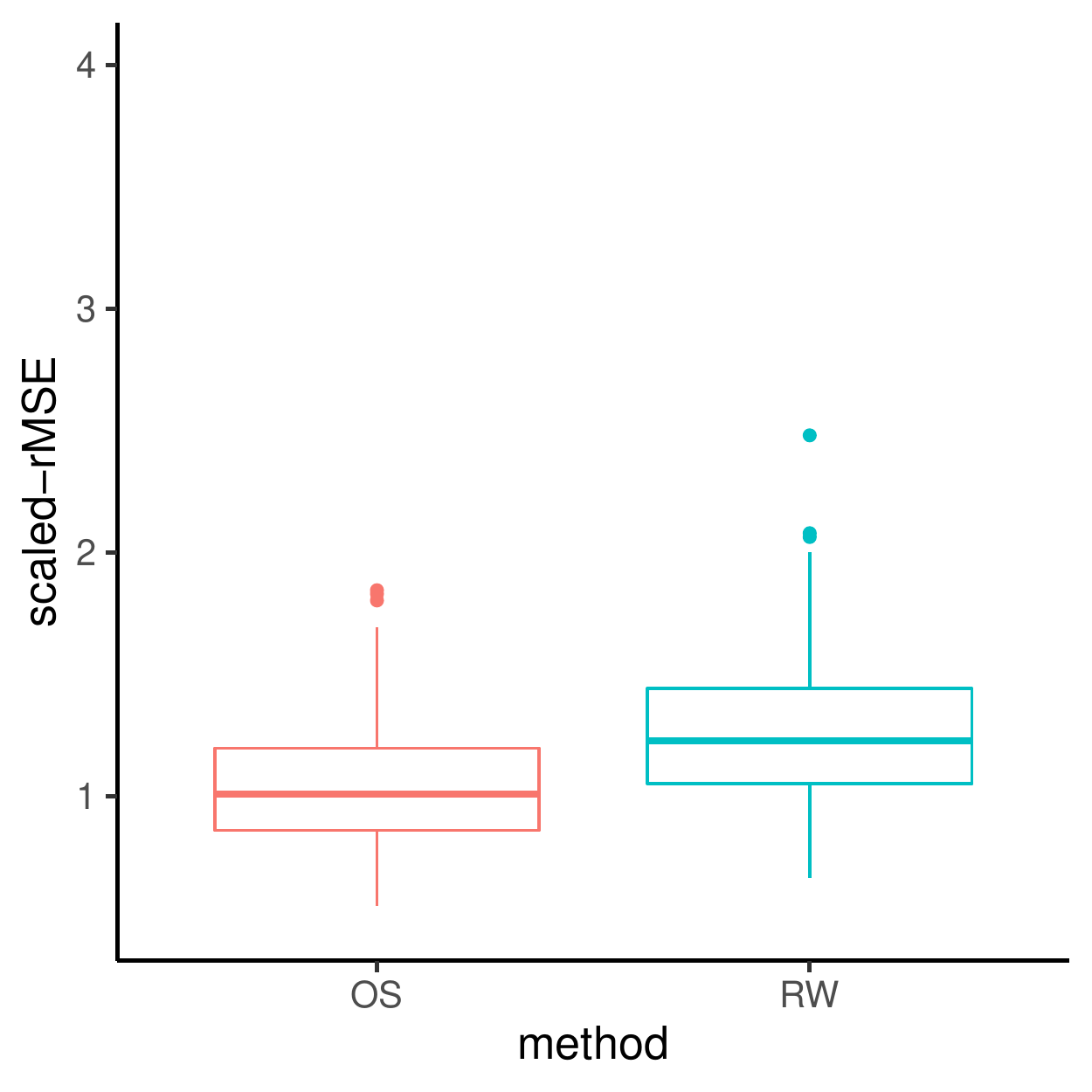}
    }
              \subfigure[rMSE ratio for $g''$]{
      \includegraphics[width=0.31\textwidth]{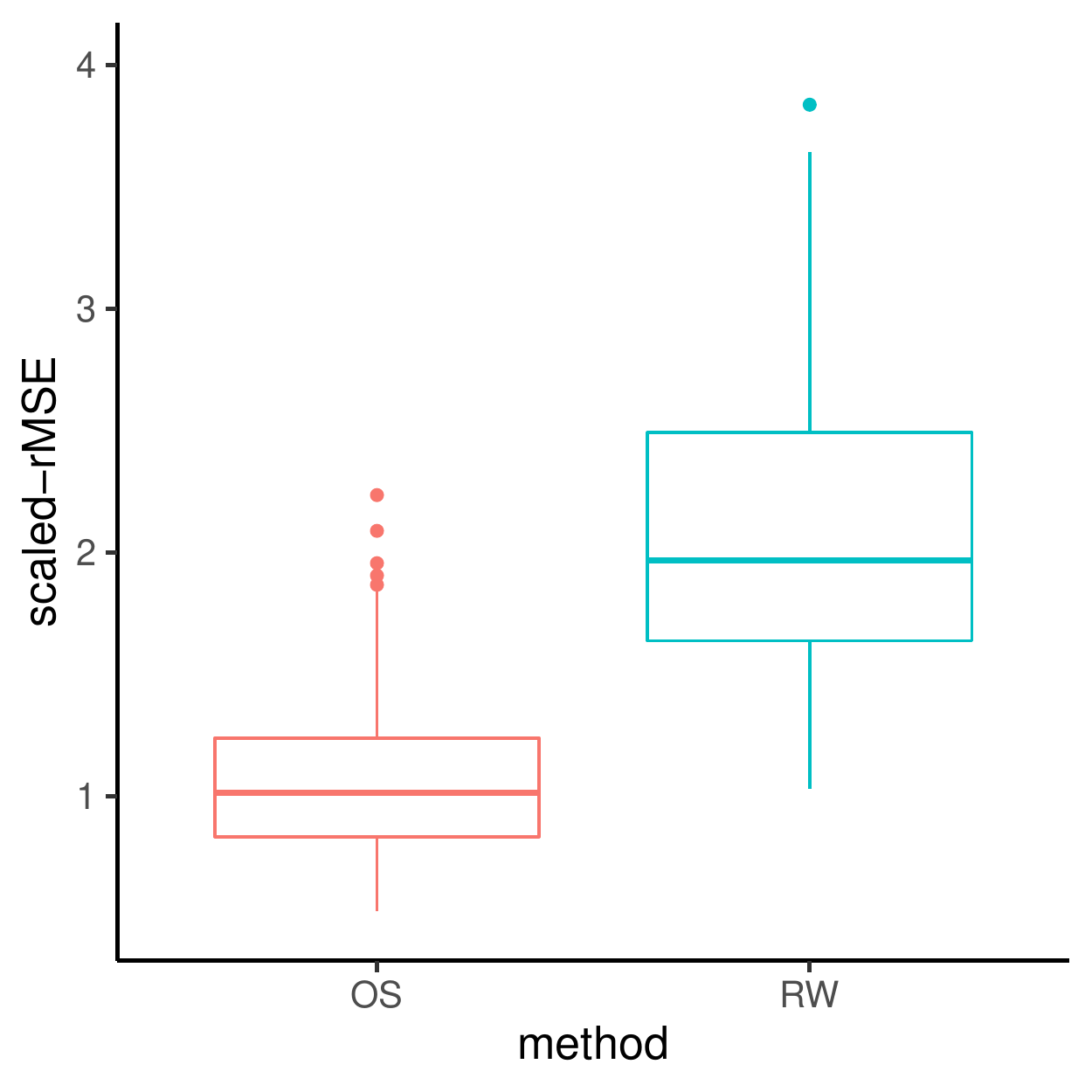}
    }
         \subfigure[Posterior mean for $g$]{
      \includegraphics[width=0.31\textwidth]{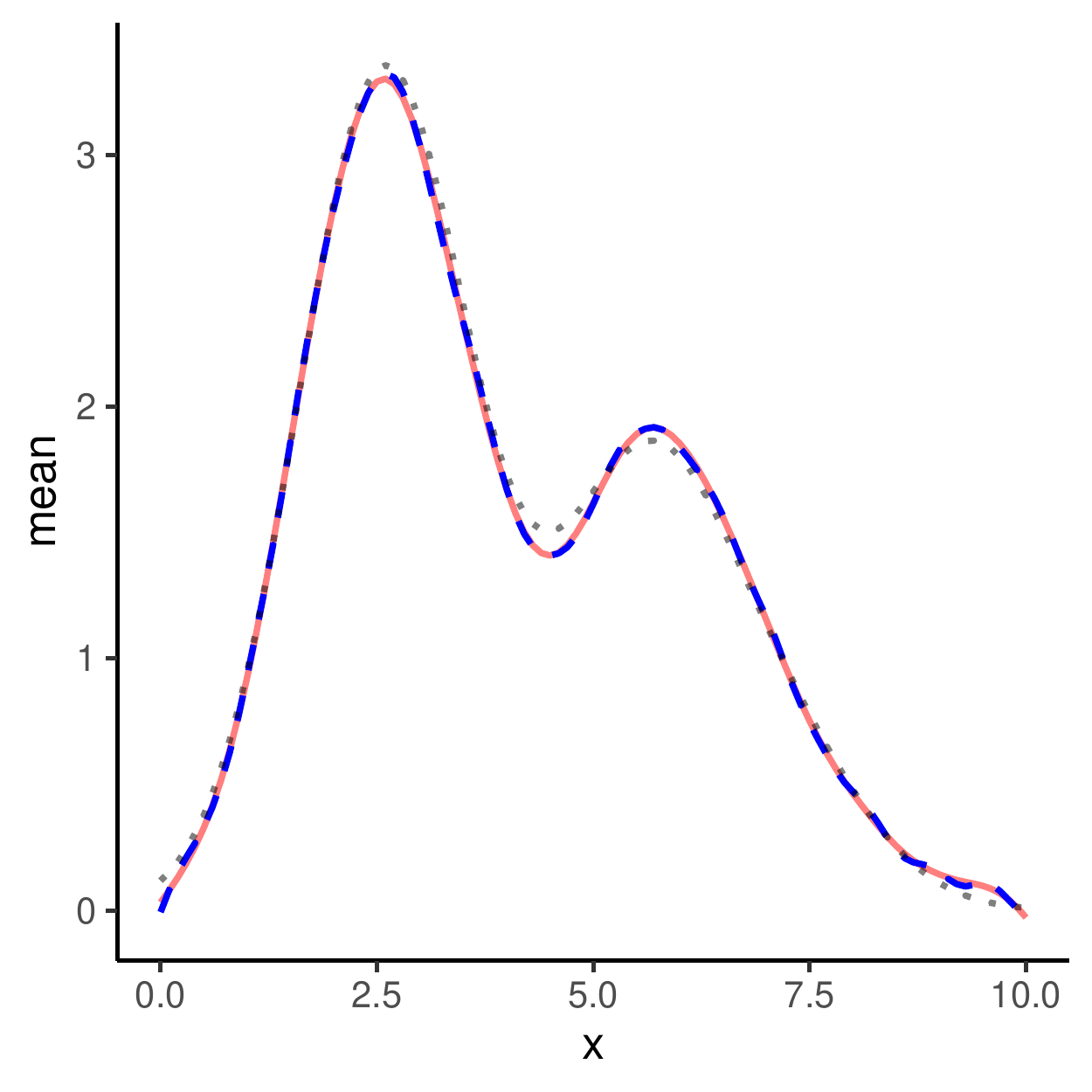}
    }
             \subfigure[Posterior mean for $g'$]{
      \includegraphics[width=0.31\textwidth]{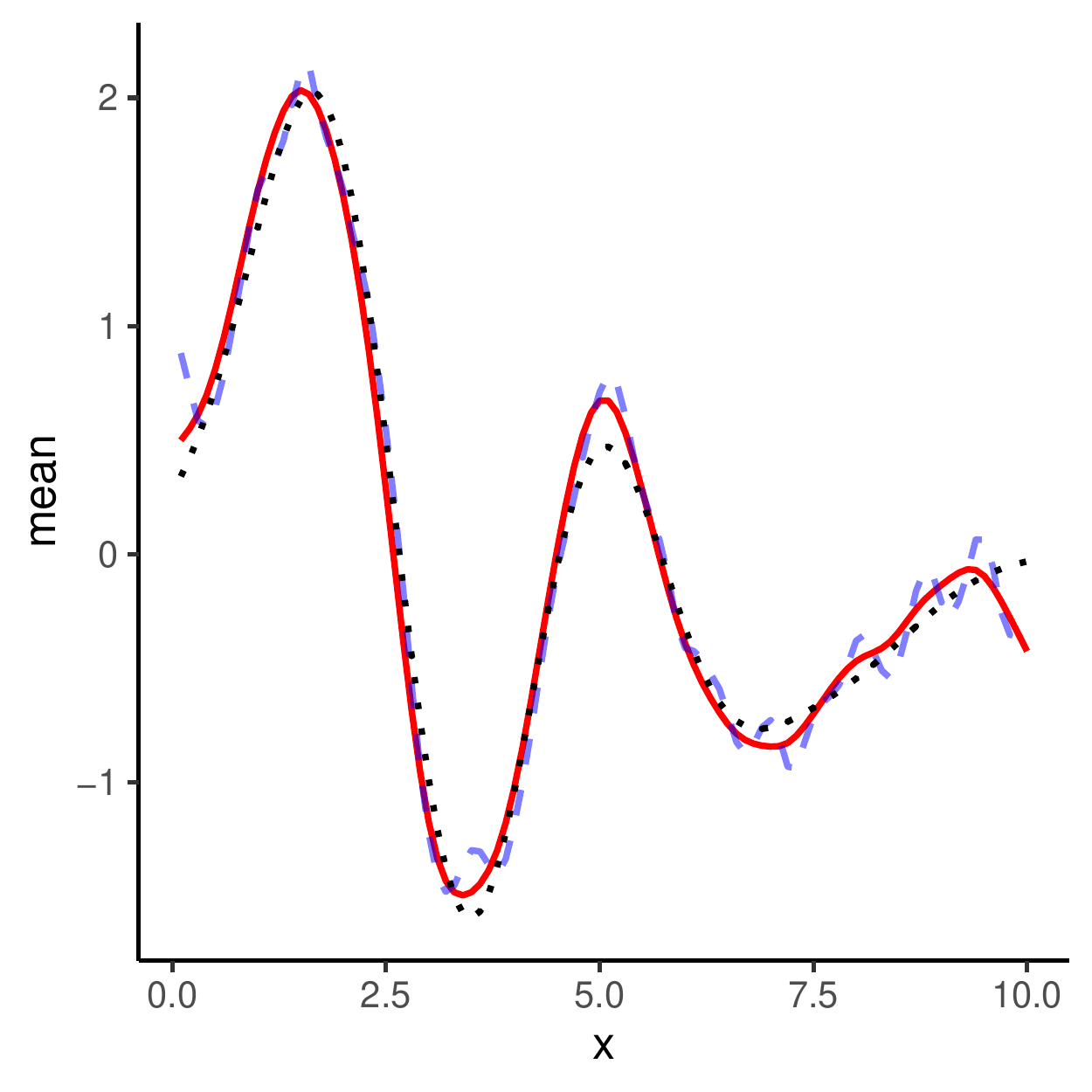}
    }
                 \subfigure[Posterior mean for $g''$]{
      \includegraphics[width=0.31\textwidth]{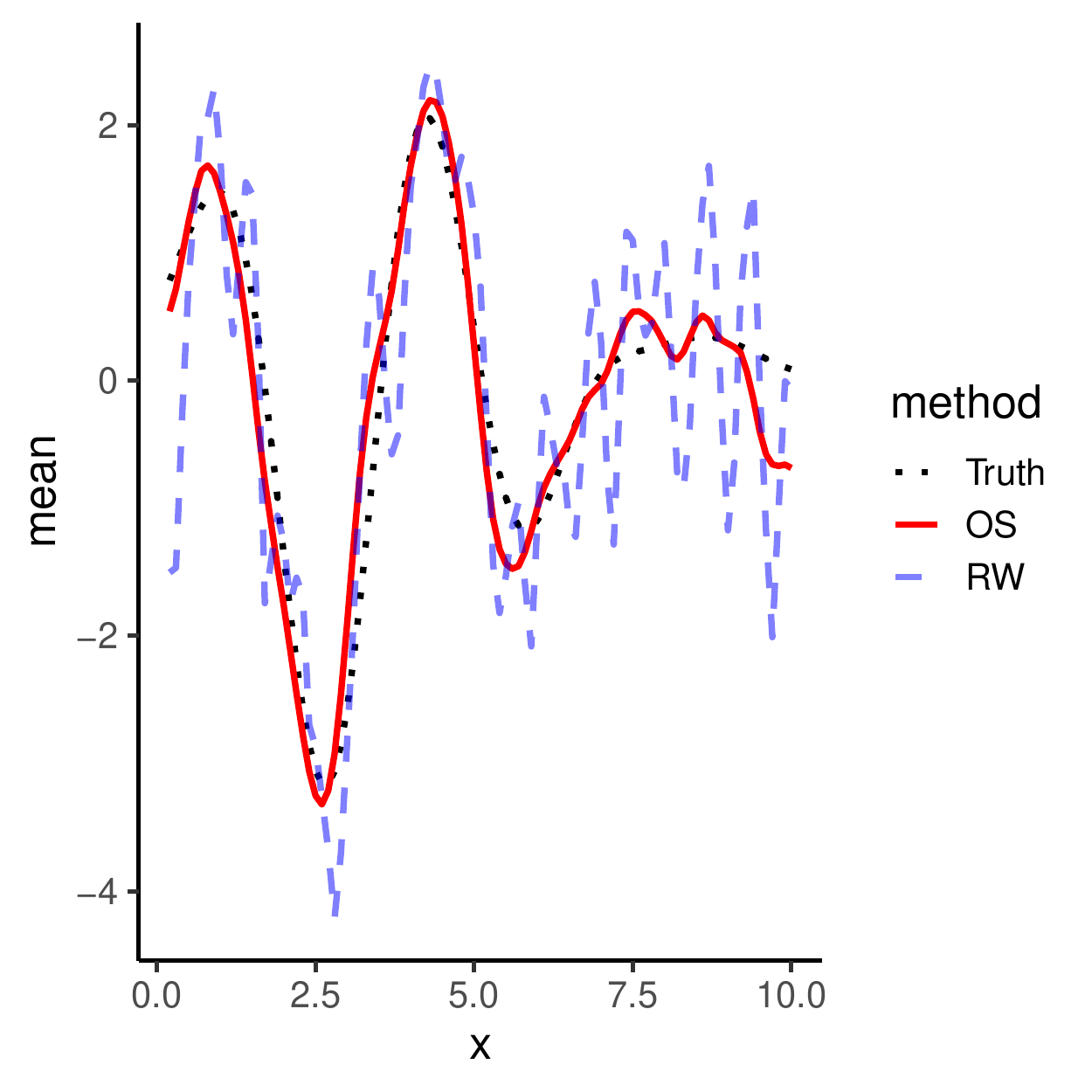}
    }
     \caption{Simulation in \cref{sec:sim2}. Figures (a)-(c) show the boxplots of rMSE scaled by the median rMSE of the O-spline method. Figures (d)-(f) show the posterior means for $g$ and its derivatives. The O-spline method is shown in red; The RW2 method is shown in blue. Although the inferential qualities for $g$ are similar, the O-spline method provides more accurate inference for the derivatives.}
    \label{fig:sim_GMM_agg}
\end{figure}


\section{Example}\label{sec:examples}

In this section, we use the O-spline method to analyze the COVID-19 daily death rates in Canada, Denmark, South Africa and South Korea from 2020-01-23 to 2022-04-26.
In this example, the model-based inference for derivatives of death rates has practical meanings to answer questions such as when was COVID death rate growing at its fastest rate, or whether COVID death rate is slowing down or speeding up. 
The data is obtained from COVID-19 Data Repository by the Center for Systems Science and Engineering (CSSE) at Johns Hopkins University \citep{dong2020interactive}. The raw data for each country is displayed in the online supplement.

Let $y_t$ denote the daily new deaths at time $t$, where $t$ denotes the time in days from 2020-03-01 up to 2022-04-26. We consider a Poisson regression model:
\begin{equation}\label{covid_mod}
    \begin{aligned}
    y_i &\sim \text{Poisson}(\exp(\eta_i)),\\
    \eta_{i} &= \boldsymbol{v}_i^T \boldsymbol{\beta} + g(x_i) + \epsilon_{i}, \\
    g(x) &\sim \text{IWP}_3(\sigma), \quad \epsilon_i \overset{iid}{\sim} N(0,\phi^2).
    \end{aligned}
\end{equation}

The model contains linear fixed effect $\boldsymbol{v}_i = (v_{i1}, ..., v_{i6})^T $ for the variable weekdays, a smoothing effect over the time variable $x$ through the unknown smooth function $g(x)$, and an observation-level random effect $\epsilon_i$ to accommodate the potential over-dispersion.
The weekdays variable $\boldsymbol{v}_i$ is coded such that $g(0)$ is interpreted as the average weekday effect. Therefore, $\boldsymbol{\beta}$ represents the additional weekday effects of Monday to Saturday relative to the average effect, and the additional weekday effect on Sunday is computed as $-\sum_{i=1}^6 \beta_i$.
Each of the linear fixed effects $\beta_i$ is given independent $N(0,100)$ prior. The over-dispersion parameter $\phi$ is modelled with an Exponential prior with a median $0.1$.

The unknown function $g(x)$ is modelled with an IWP$_3(\sigma)$ prior in \cref{eqn:expansion}, with $\gamma_l \overset{iid}{\sim} N(0,1/100)$ for $0\leq l < 3$.
We then approximate the IWP prior using our proposed O-spline method with $k = 100$ to balance the computational efficiency and the approximation accuracy. 
For the parameter $\sigma$, we assign an Exponential prior on $\sigma(h)$ with median $\log(2)$, where $h$ is taken to be 7 days. This prior can be interpreted as with a roughly 50 percent chance that the death rate could be scaled by $4$ or $1/4$ in a week.

The inferential results for the function $\exp [g(x)]$ which denotes the evolution of COVID death rate in each country after adjusting for the weekday effect and population size are shown in Figure \ref{fig:covidDeath}. 
Model-based inference is also done for the derivative of the COVID death rate $g'(x)\exp [g(x)]$ with adjusted results shown in Figure \ref{fig:covidDeriv}. 
It can be observed from Figure \ref{fig:covidDeriv} that COVID death rate was growing fastest in South Korea during the waves around March 2022, whereas in Canada the fastest time was between March 2020 and July 2020. 
For South Africa, the method suggests that the COVID death increased at fastest speed around the end of 2020. In Denmark, the three waves had similar speeds at their peaks. 

The posteriors of weekday effects in each country as well as the posteriors of the overdispersion and the PSD can be found in Figure \ref{fig:covidOther}. The inferential results for the log relative risk $g(x)$ and its derivative $g'(x)$ are also provided in the online supplement.

\begin{figure}[!p]
    \centering
             \subfigure[Canada]{
      \includegraphics[width=0.41\textwidth]{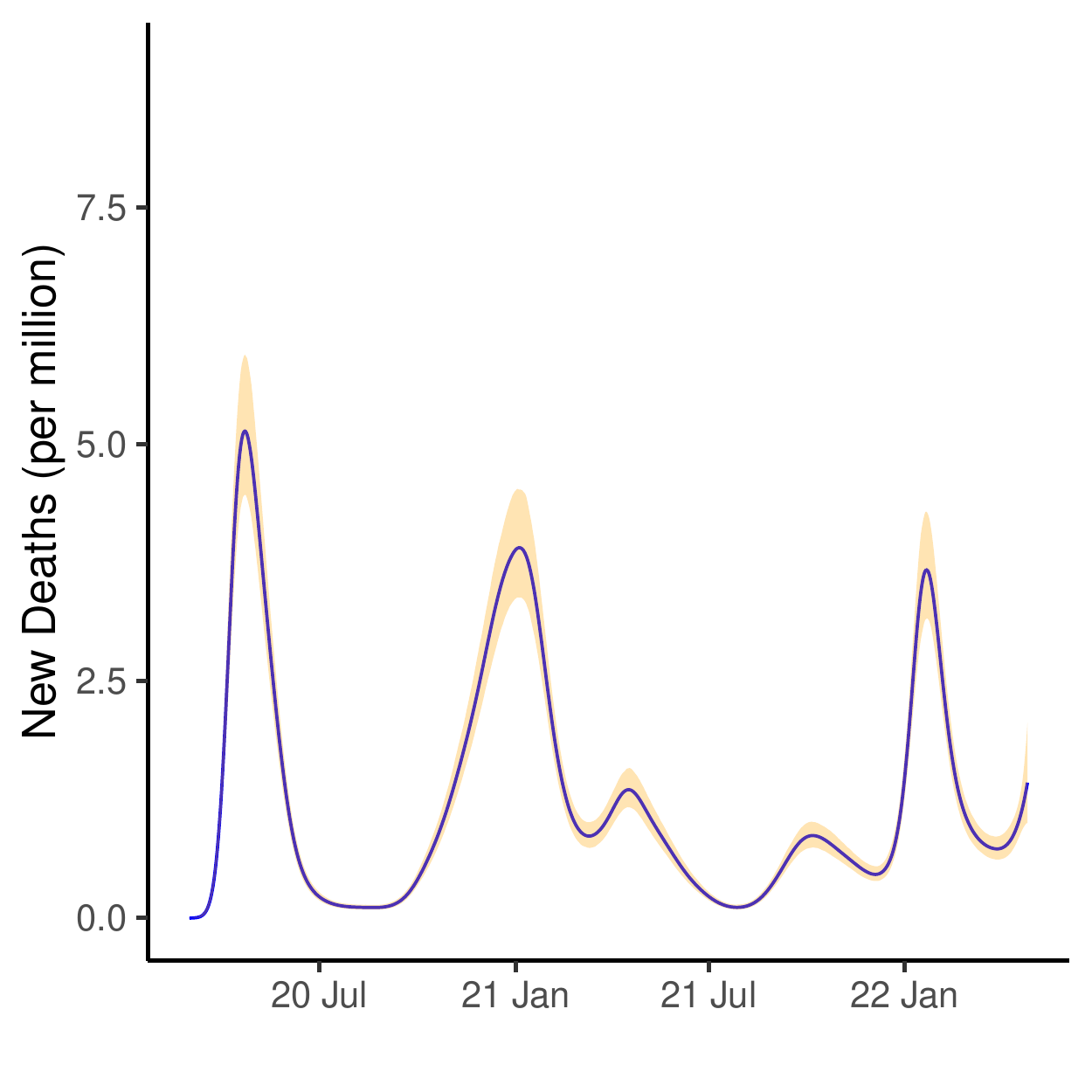}
    }
             \subfigure[Denmark]{
      \includegraphics[width=0.41\textwidth]{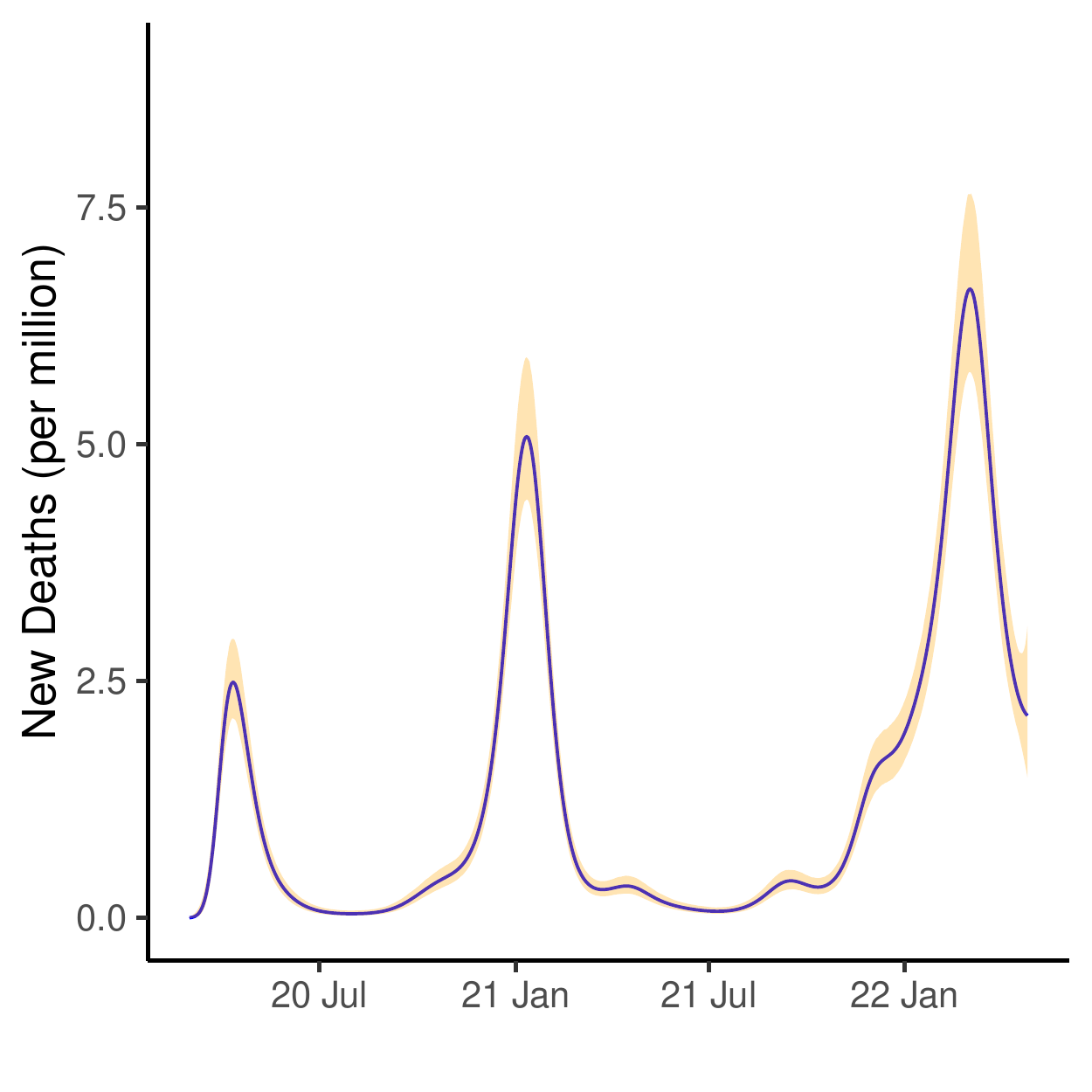}
    }
              \subfigure[South Africa]{
      \includegraphics[width=0.41\textwidth]{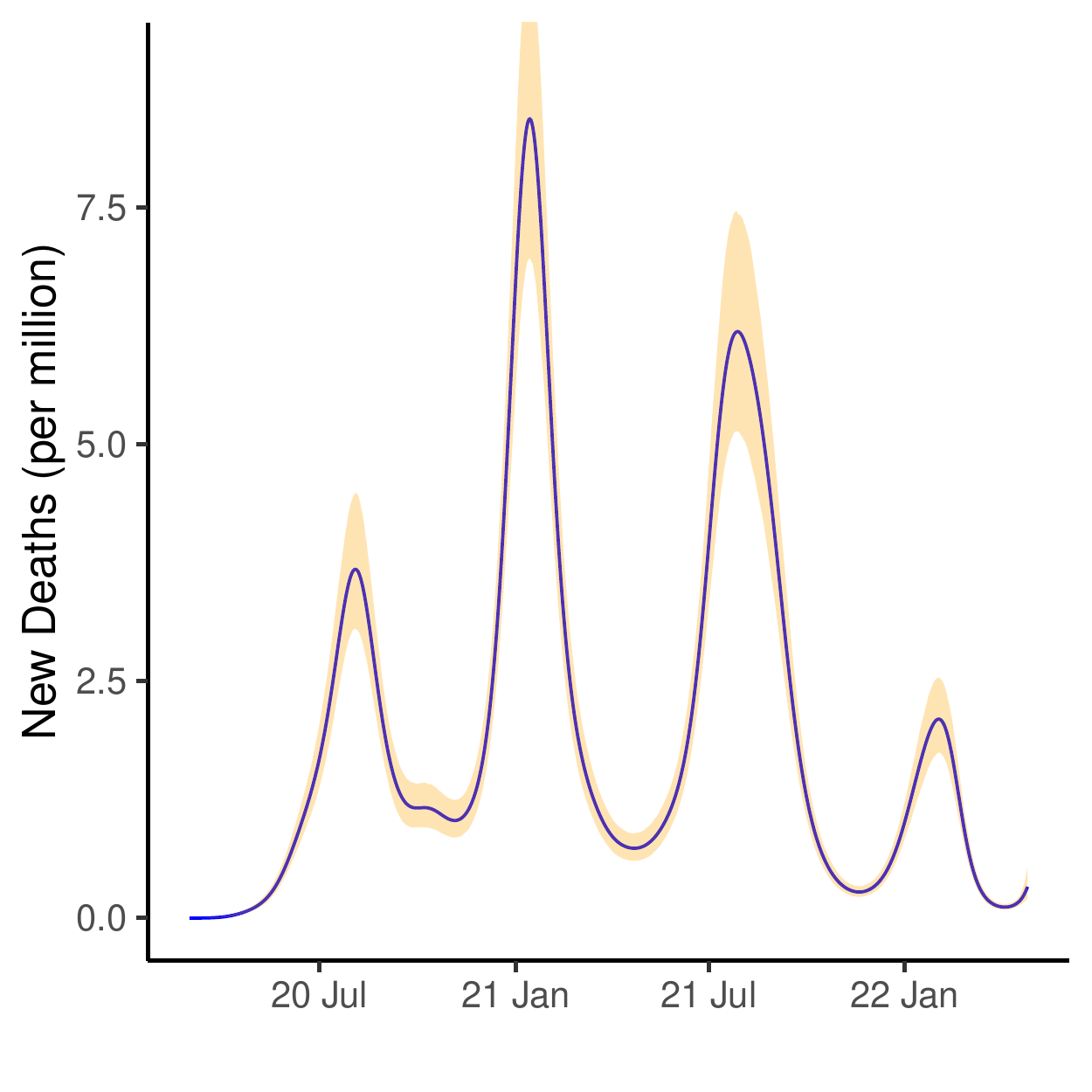}
    }
               \subfigure[South Korea]{
      \includegraphics[width=0.41\textwidth]{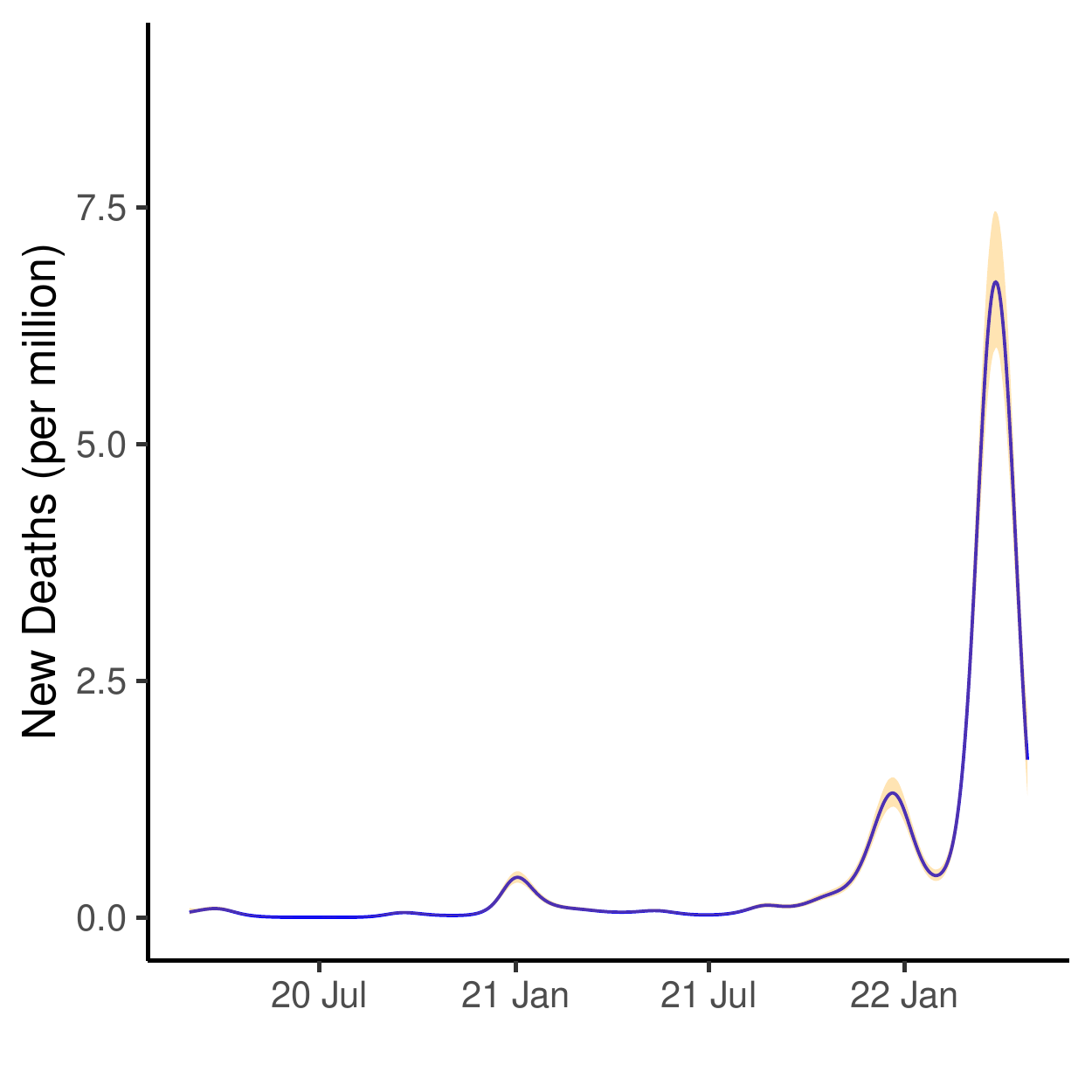}
    }
     \caption{Results for the COVID data analysis in \cref{sec:examples}. In figures (a)-(d), the blue line is the posterior mean of $\exp[g(x)]$ the O-spline (OS) method; and the orange range is the 95 \% pointwise posterior credible interval of the OS method. The death counts were adjusted based on the population size (per million) of each country at 2020. }
    \label{fig:covidDeath}
\end{figure}

\begin{figure}[!p]
    \centering
             \subfigure[Canada]{
      \includegraphics[width=0.41\textwidth]{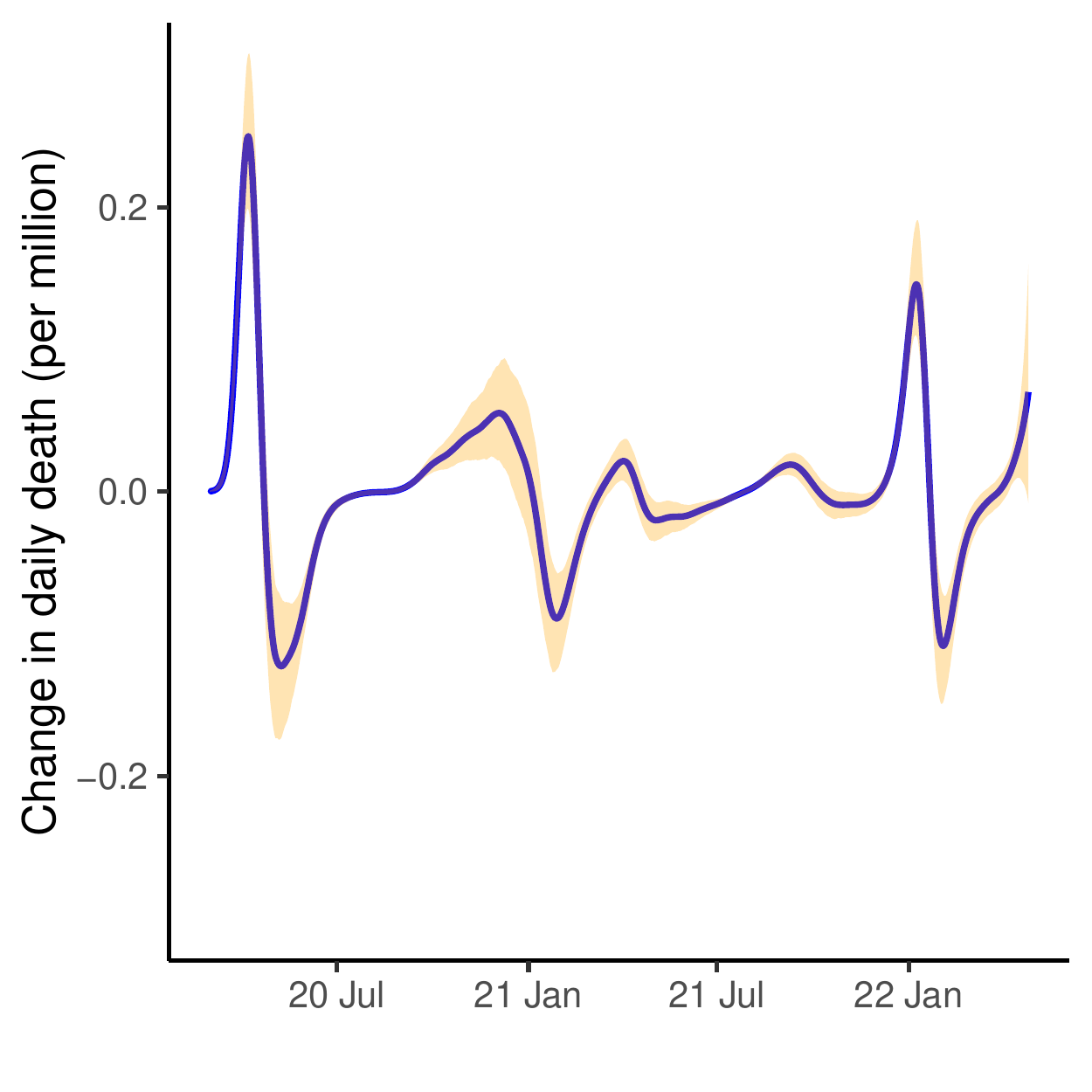}
    }
             \subfigure[Denmark]{
      \includegraphics[width=0.41\textwidth]{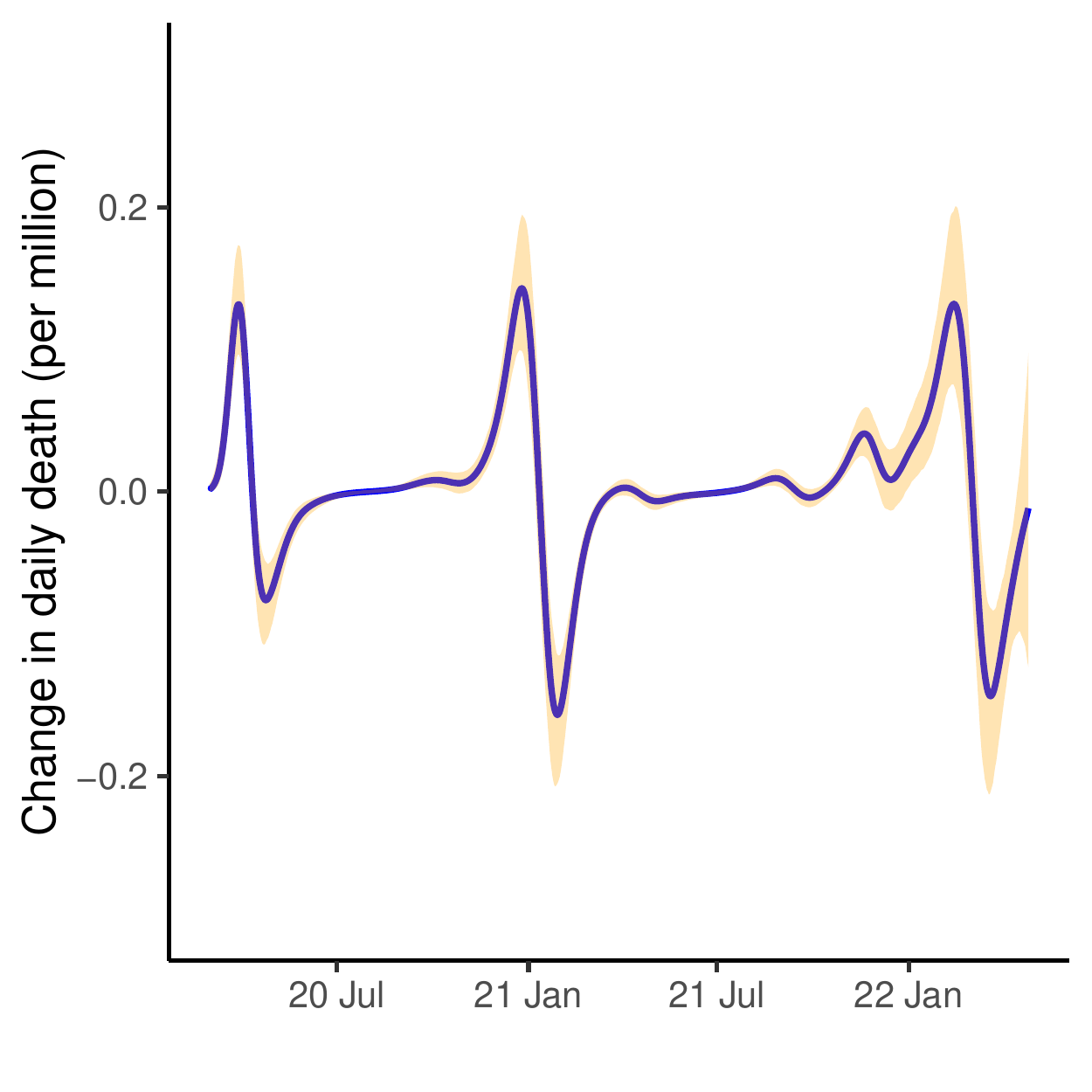}
    }
              \subfigure[South Africa]{
      \includegraphics[width=0.41\textwidth]{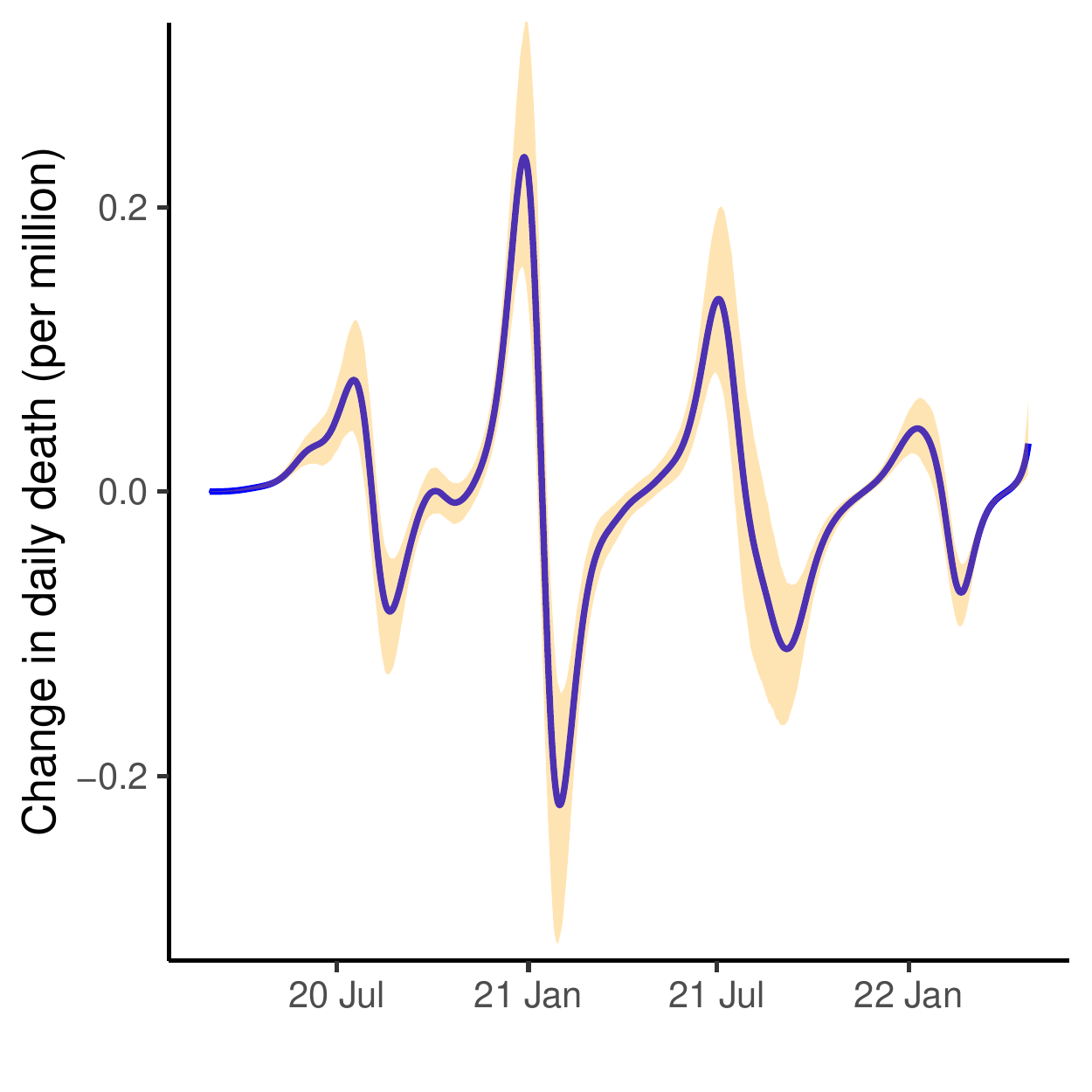}
    }
              \subfigure[South Korea]{
      \includegraphics[width=0.41\textwidth]{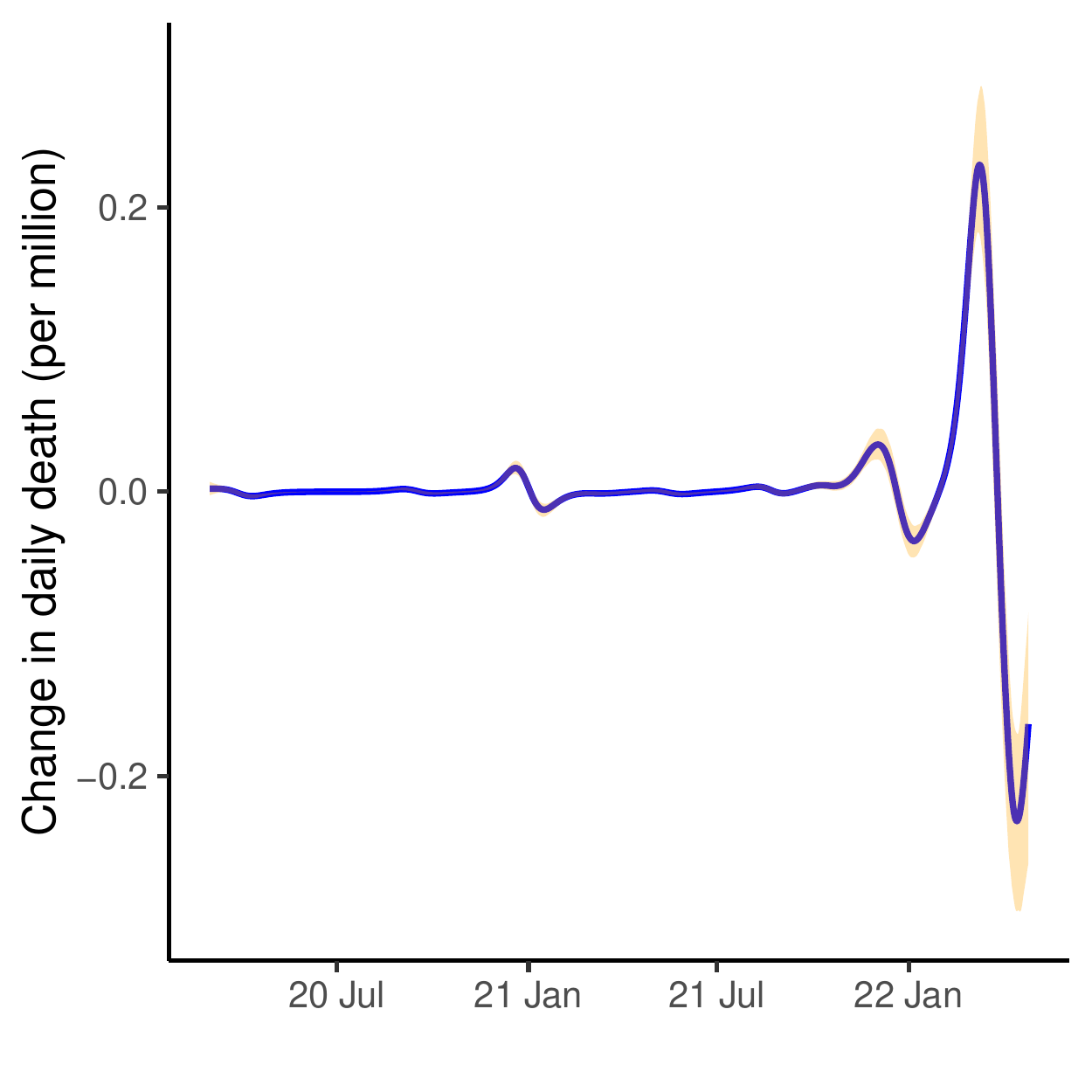}
    }
     \caption{Results for the COVID data analysis in \cref{sec:examples}. In figures (a)-(d) the posterior mean of the derivative $g'(x) \exp [g(x)]$ using the OS method is shown in blue; the orange range is the 95 \% pointwise posterior credible interval of the OS method. The results were adjusted based on the population size (per million) of each country at 2020.}
    \label{fig:covidDeriv}
\end{figure}

\begin{figure}[!p]
    \centering
                 \subfigure[Weekday Effect]{
      \includegraphics[width=0.41\textwidth]{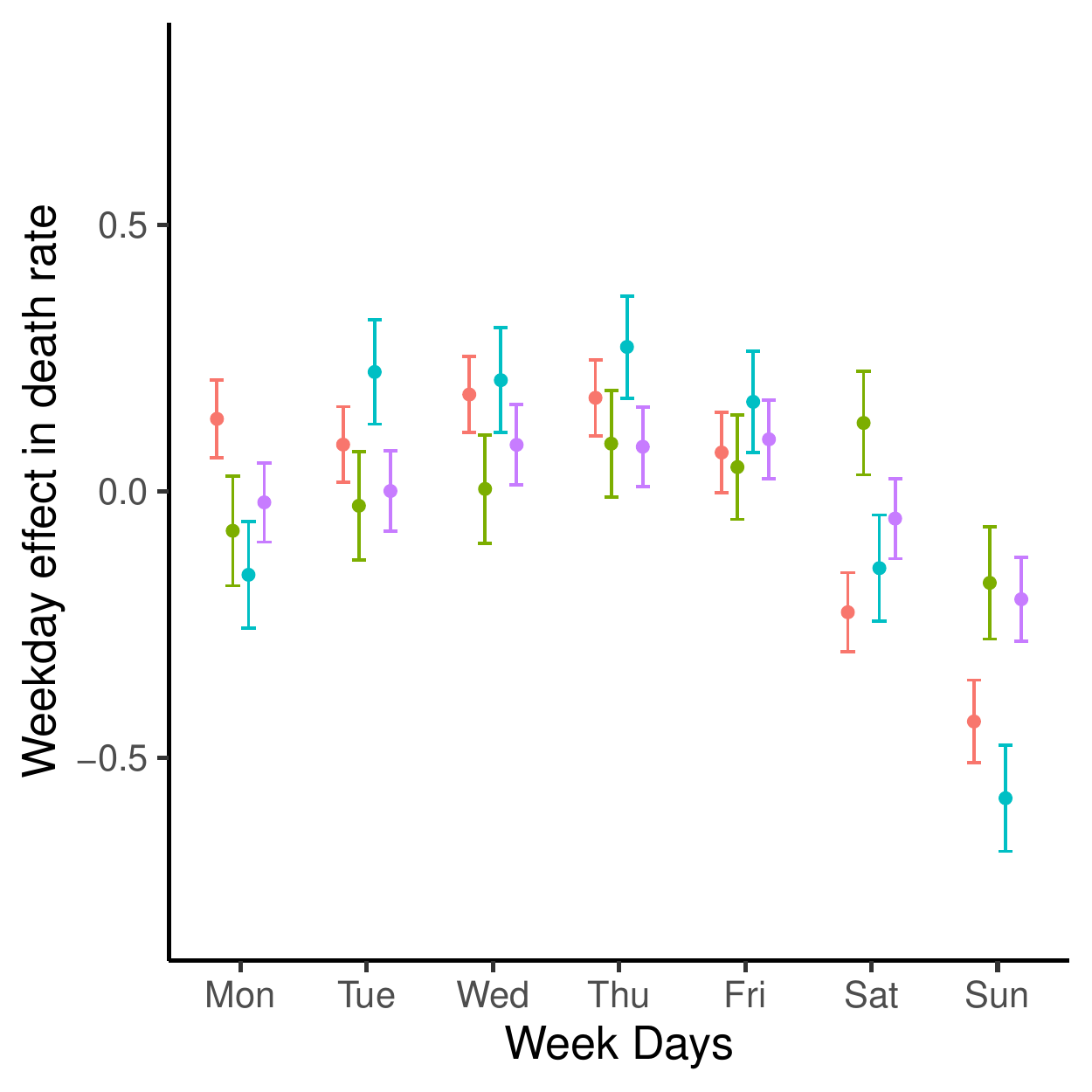}
    }
             \subfigure[7-days PSD]{
      \includegraphics[width=0.41\textwidth]{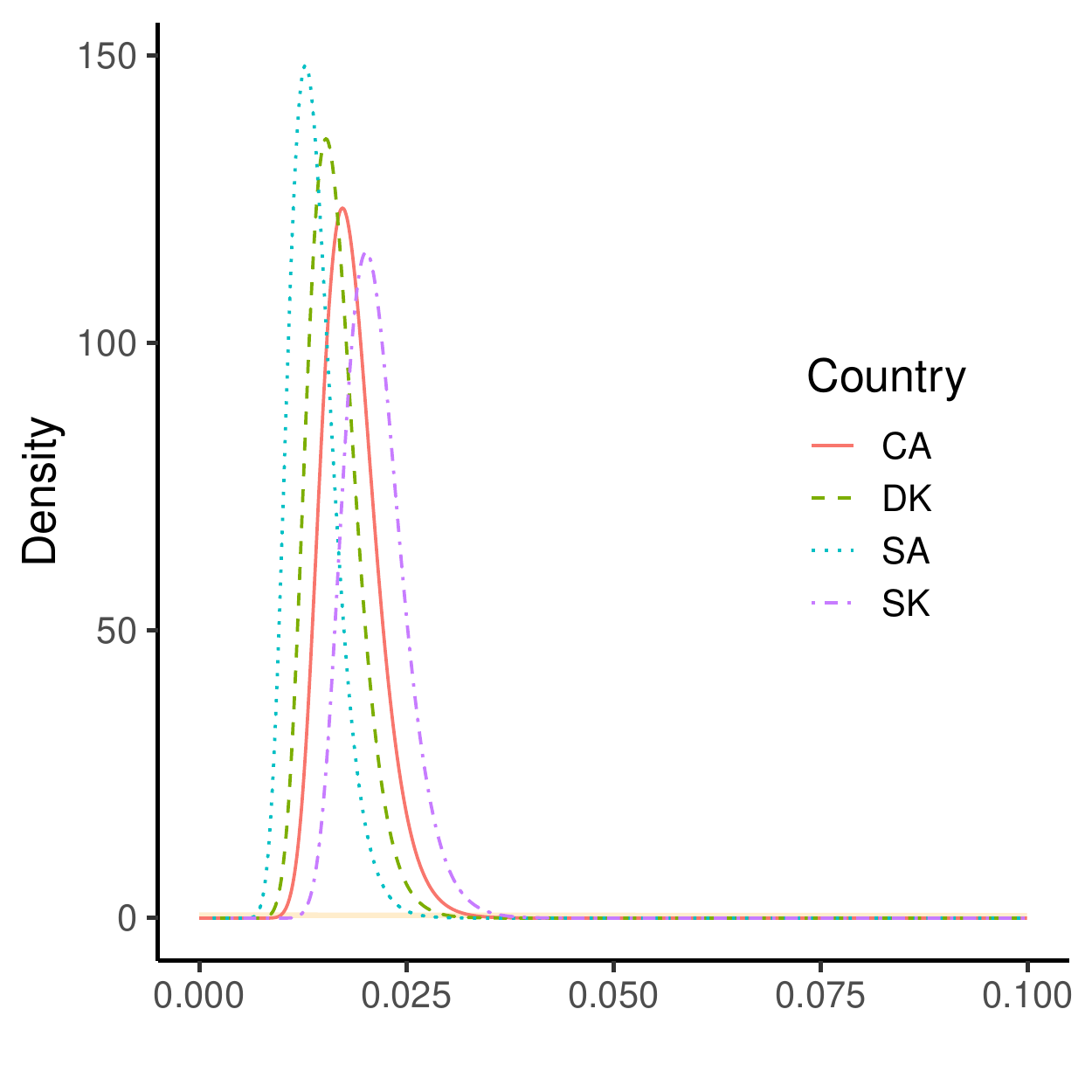}
    }
             \subfigure[Overdispersion $\phi$]{
      \includegraphics[width=0.41\textwidth]{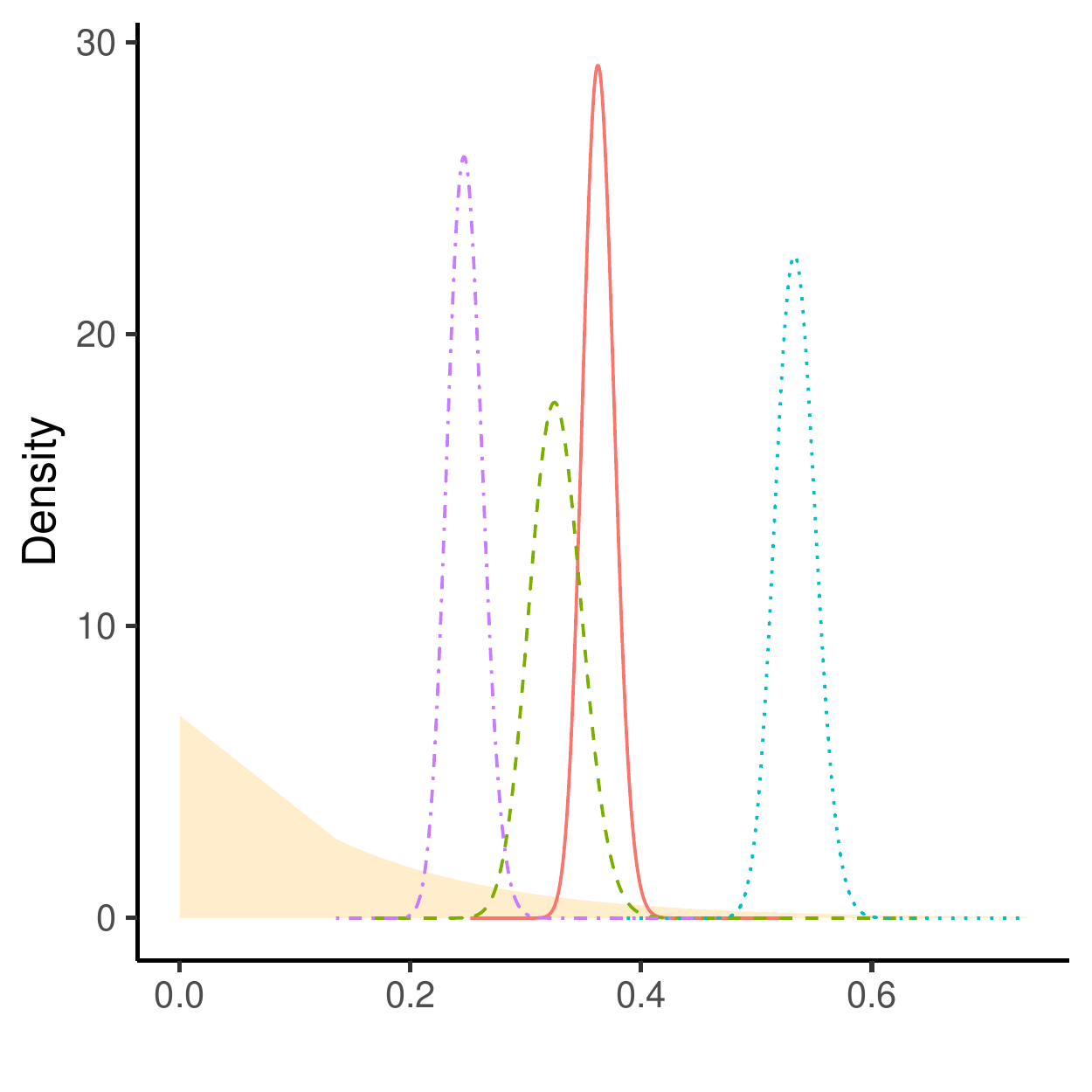}
    }
     \caption{Results for the COVID data analysis in \cref{sec:examples}. Figure (a) shows the inference of weekday effects in each country. The solid dots denote the posterior mean of each $\beta_i$, the weekday effect relative to the average effect. The lines denote the 95 \% credible interval. Figure (b) shows the posterior density of the 7-days PSD $\sigma(7)$ in each country, and (c) shows the posterior density of the overdispersion parameter $\phi$ in each country. For both figures, the orange shaded area represents the prior distribution.}
    \label{fig:covidOther}
\end{figure}


\section{Discussion}\label{sec:discussion}

In this paper, we considered model-based smoothing method with Integrated Wiener's process (IWP), and provided a novel finite dimensional O-spline approximation that works for any order of IWP. 
The proposed approximation is able to give consistent inference for all the derivatives of the function, and we prove its convergence to the true IWP both theoretically and practically with the simulation study. 
To select and interpret prior for the parameter $\sigma$, we introduced a prior elicitation approach based on the notion of $h$-units PSD $\sigma(h)$, which unlike the existing approach based on the marginal SD of the approximation \citep{scalingigmrf}, is defined based on the exact process and hence invariant to the choice of knots.
The utility of the method has been illustrated both by the simulation and the data analysis example with COVID death rates in four different countries.

The order of the IWP$_p(\sigma)$ model implicitly states a priori knowledge one has on the function's smoothness, measured by the number of times the unknown function can be continuously differentiated. If the choice of the order $p$ is not clear by the priori knowledge, one can consider selecting it using methods such as the Bayes factor since marginal likelihood is well-defined within the proposed approximation.

Although the proposed method is already flexible in terms of the supported IWP order and the inference for the derivative, it can be further improved and generalized to allow adaptive smoothing with a varying smoothing parameter $\sigma$, see \cite{adaptivesmoothingsplines} for an example. 

For the simplicity of our presentation, we assumed by default that the IWP starts at the leftmost point of the region of interest. However, the same procedure can still be applied to other starting values using two independent IWP moving toward different directions. This could further simplify the computation in certain scenarios by introducing a more sparse design matrix.

Another potential direction for improvement is for multivariate smoothing such as the spatial setting studied in \cite{spde}, where the generalization of the proposed O-spline basis from univariate space will not be trivial. We leave these to the future work.

\section*{Disclosure statement}

The authors report there are no competing interests to declare.
\section*{Supplemental Materials}

The proofs of \cref{lem:convergence} and \cref{thrm:convergence-joint} are provided in the Appendix. The details of our FEM procedure and additional figures for \cref{sec:examples} are shown in the online supplement.

The codes to reproduce all the results in the main paper are provided at the online repository \href{https://github.com/AgueroZZ/Smooth_IWP_code}{github.com/AgueroZZ/Smooth_IWP_code}.

\bigskip

\bibliographystyle{chicago}

\bibliography{bibliography}

\newpage

\appendix

\section*{Appendix: Proof Of Main Results}

\begin{manuallemma}{1}[Convergence of O-spline Approximation]
Let $\Omega = [a,b]$ where $a,b \in \mathbb{R}^+$ and let $g \sim \text{IWP}_p(\sigma)$ with $p \in \mathbb{Z}^+$. Assume the knots $\{s_1, ..., s_k\}$ are equally spaced over $\Omega$ for each $k \in \mathbb{N}$, and $\tilde{g}_k$ denotes the corresponding $p$th order (O-spline) approximation defined as in \cref{equ:FEM-approxi}, then:
$$||\C - \C_{k}||_{\infty} = O(1/k),$$
where $\C(s,t) = \Cov[g(s),g(t)]$, $\C_{k}(s,t) = \Cov[\tilde{g}_k(s),\tilde{g}_k(t)]$ and $||\C - \C_{k}||_{\infty} \equiv \sup_{s,t \in \Omega}|\C(s,t) - \C_{k}(s,t)|$.
\end{manuallemma}

\begin{proof}
To prove this theorem to general integer order $p$, we will start with the case when $p = 1$. Assume without the loss of generality that $\gamma_i = 0$ for each $i$ in \cref{eqn:expansion}, and $\sigma = 1$. 
Also, assume without the loss of generality that $\Omega = [0,1]$ and hence $s_i = i/k$ for each i. Then the true covariance function of IWP in this case will be $$\C(s,t) = \min\{s,t\}.$$

To compute the covariance of the approximation, we have the following:
\begin{equation}
    \begin{aligned}
    \C_{k}(s,t) &:= \Cov \bigg(\tilde{g}_k(s), \tilde{g}_k(t) \bigg) \\
                 &= \Cov \bigg(\sum_{i=1}^k w_i \varphi_i(s), \sum_{i=1}^k w_i \varphi_i(t) \bigg) \\
                 &= k \Phi(s)^T \Phi(t),\\
    \end{aligned}
\end{equation}
where $\Phi(t) = (\varphi_1(t), ..., \varphi_k(t))^T$ and $\Phi(s) = (\varphi_1(s), ..., \varphi_k(s))^T$, because of the distribution of $w_i \overset{iid}{\sim} N(0,k)$ according to our definition.

Now assume that $s \leq t$ without loss of generality. Because of the overlapping property of the O spline basis $\varphi_i$, it is obvious that only knots located in the region $[0,s]$ will have basis functions with non-trivial contributions to the above inner product. Therefore,
\begin{equation}
    \begin{aligned}
   k \Phi(s)^T \Phi(t) &= k\sum_{i=1}^{\lfloor ks \rfloor} \varphi_i(s)\varphi_i(t) \\
   &= k \sum_{i=1}^{\lfloor ks \rfloor} \frac{1}{k^2} \ \text{since $p$ = 1} \\
   &= \lfloor ks \rfloor/k,
    \end{aligned}
\end{equation}
since $\phi_i(s) = \phi_i(t) = \frac{1}{k}$ for $i \leq \lfloor ks \rfloor$.

\ 

\noindent For any fixed $t \in \Omega$, we have 
\begin{equation}
    \begin{aligned}
   \sup_{s \in \Omega} |\C(s,t) - \C_{k}(s,t)| &\leq \sup_{s \in [0,t]} |\C(s,t) - \C_{k}(s,t)| + \sup_{s \in [t,1]} |\C(s,t) - \C_{k}(s,t)|\\
   &= \sup_{s \in [0,t]} |s - \lfloor ks \rfloor/k| + \sup_{s \in [t,1]} |t - \lfloor kt \rfloor/k| \\
   &\leq 2/k.
    \end{aligned}
\end{equation}
Since $t$ is arbitrary, this implies $\sup_{s,t \in \Omega}|\C(s,t) - \C_{k}(s,t)| \leq 2/k$ and hence we prove the case for $p=1$.

To generalize the result to the higher order of $p$, we need the following proposition about the Gaussian process:

\begin{proposition}[Integration of Gaussian Process]\label{prop:GPderiv}
Assume $f(t)$ is a Gaussian process with continuous sample paths for $t \in \Omega := [0,1]$ with the boundary condition $f(0) = 0$, then $g(s) := \int_0^s f(t) dt$ is still a Gaussian process, and its covariance function at $s^*, t^* \in \Omega$ can be computed as $$\C_g(s^*, t^*) = \int_{0}^{s^*} \int_{0}^{t^*} \C_f(s, t) dt ds$$ where $\C_f(.)$ is the covariance function of $f$.
\end{proposition}
\noindent This proposition follows from section 2.4.3 of \cite{adler2010geometry}. 

Now, consider the case where $p=2$, then because of the derivative consistency property of the overlapping spline basis, it is immediate that $\tilde{g}^{(1)}_k(x)$ is the first order O-spline approximation for the first order IWP ${g}^{(1)}(x)$, whose convergence is already established in the proof above.
Furthermore, for $p=1$, the proposed approximation has continuous sample path and covariance function $\C_{k}(s,t) = \lfloor ks \rfloor/k \leq 1$ as shown above.

Let $\C^{[p]}(s,t)$ and $\C^{[p]}_k(s,t)$ denote the covariance of the $p$th order IWP and its approximation, then for any choice of $s,t \in \Omega$, we have:
\begin{equation}
    \begin{aligned}
   \sup_{s,t \in \Omega} &|\C^{[2]}(s,t) - \C^{[2]}_{k}(s,t)|\\
   &= \sup_{s,t \in \Omega} \bigg|\int_{0}^{s} \int_{0}^{t} \C^{[1]}(s^*,t^*) - \C^{[1]}_{k}(s^*,t^*) dt^* ds^*\bigg| \\
   &\leq \sup_{s,t \in \Omega} \int_{0}^{s} \int_{0}^{t} \bigg| \C^{[1]}(s^*,t^*) - \C^{[1]}_{k}(s^*,t^*) \bigg| dt^* ds^* \\
   &\leq \sup_{s,t \in \Omega} \int_{0}^{1} \int_{0}^{1} \bigg| \C^{[1]}(s^*,t^*) - \C^{[1]}_{k}(s^*,t^*) \bigg| dt^* ds^* \\
   & \leq \int_{0}^{1} \int_{0}^{1} \frac{2}{k} dt^* ds^* \ \text{by previous result on $p=1$} \\
   &= \frac{2}{k}.
    \end{aligned}
\end{equation}
So the convergence of the covariance function for $p=2$ is established. Furthermore, because the region $[0,s] \times [0,t] \subseteq \Omega \times \Omega$ is compact, the covariance function for higher order $p$ will still be bounded by $2/k$. Hence the convergence result can be generalized to any positive integer $p$ by induction.

\end{proof}

\begin{manualtheorem}{1}[Main Theorem]
Given the same setting and notations as in \cref{lem:convergence},  for any non-negative integers $q_1 \leq p-1$ and $q_2 \leq p-1$:
$$||\C^{(q_1,q_2)} - \C^{(q_1,q_2)}_{k}||_{\infty} = O(1/k),$$
where $\C^{(q_1,q_2)}(s,t) = \Cov[g^{(q_1)}(s), g^{(q_2)}(t)]$ and $\C^{(q_1,q_2)}_{k}(s,t) = \Cov[\tilde{g}^{(q_1)}_k(s), \tilde{g}^{(q_2)}_k(t)]$.
\end{manualtheorem}

\begin{proof}

To prove the general convergence result from \cref{thrm:convergence-joint}, we start with proving a special case given as the following \cref{lem:cross-convergence}:

\begin{manuallemma}{2}[Convergence of Cross-Covariance]\label{lem:cross-convergence}
Let $\Omega = [a,b]$ where $a,b \in \mathbb{R}^+$ and let $p \geq 1$ and $q \leq p-1$ be arbitrary positive integers. Let $g \sim \text{IWP}_p(\sigma)$. Assume the knots $\{s_1, ..., s_k\}$ are equally spaced over $\Omega$ for each $k \in \mathbb{N}$, and $\tilde{g}_k(x)$ denotes the corresponding $p$th order (O-spline) approximation defined in \cref{equ:FEM-approxi}, then:
$$||\C^{(0,q)} - \C^{(0,q)}_{k}||_{\infty} = O(1/k),$$
where $\C^{(0,q)}(s,t) = \Cov[g(s), g^{(q)}(t)]$ and $\C^{(0,q)}_{k}(s,t) = \Cov[\tilde{g}_k(s), \tilde{g}_k^{(q)}(t)]$.
\end{manuallemma}

For ease of notation, define the differentiation and integration operators as $$D^p_t (g(s,t)) := \frac{\partial^p g}{\partial t^p} (s,t),$$ and $$I^p_t(g(s,t)) := \int_{0}^{t} \int_{0}^{t_1} ... \int_{0}^{t_{p-1}} g(s,t_{p}) dt_{p} dt_{p-1} ... dt_1.$$
Both operators are linear. Again for simplicity, we consider $\Omega = [0,1]$ without the loss of generality. When $p=1$ the theorem is trivial, so we consider the case where $p>1$.

Let $s,t \in \Omega$ be fixed and let $g$ follows the $p$th order IWP with $\tilde{g}_k$ denotes its approximation. Let $\C^{[p]}(s,t)$ denotes its auto-covariance and $\C^{[p]}_k(s,t)$ denotes the auto-covariance of the approximation.

Following from the proof of \cref{lem:convergence}, $\{\C^{[p]}_k(s,t)\}_k$ are uniformly bounded. Applying the Fubini's theorem with result of \cref{prop:GPderiv}, we get:
\begin{equation}
    \begin{aligned}
       \C^{(0,q)}_k(s,t) &= \mathbb{E}[\tilde{g}_k(s) \frac{\partial^q}{\partial t^q} \tilde{g}_k(t)]\\
       &=  \frac{\partial^q}{\partial t^q} \mathbb{E}[\tilde{g}_k(s) \tilde{g}_k(t)] \\
       &= D^q_t \C^{[p]}_k(s,t) \\
       &= D^q_t I^{p-1}_t I^{p-1}_s [\C_k^{[1]}(s,t)] \\
       &= I^{p-q-1}_t I^{p-1}_s [\C_k^{[1]}(s,t)].
    \end{aligned}
\end{equation}
Similarly, $\C^{(0,q)}(s,t)$ can be written as:
\begin{equation}
    \begin{aligned}
       \C^{(0,q)}(s,t) &= D^q_t \C^{[q]}(s,t) \\
       &=  I^{p-q-1}_t I^{p-1}_s [\C^{[1]}(s,t)].
    \end{aligned}
\end{equation}
Then the following upper bound can be achieved with the result of \cref{lem:convergence}:
\begin{equation}
    \begin{aligned}
       \sup_{s,t \in \Omega} |\C^{(0,q)}(s,t) - \C^{(0,q)}_k(s,t)| &= \sup_{s,t \in \Omega} |I^{p-q-1}_t I^{p-1}_s \C^{[1]}(s,t) - \C_k^{[1]}(s,t)| \\
       &\leq \sup_{s,t \in \Omega} I^{p-q-1}_t I^{p-1}_s |\C^{[1]}(s,t) - \C_k^{[1]}(s,t)| \\
       &\leq  \sup_{s,t \in \Omega} I^{p-q-1}_t I^{p-1}_s \frac{2}{k} \\
       &\leq \frac{2}{k},
    \end{aligned}
\end{equation}
where the last inequality follows from the compactness of $\Omega$. This concludes the proof of \cref{lem:cross-convergence}.

With \cref{lem:cross-convergence} and the property of the O-spline approximation in \cref{equ:sim-property-OS}, \cref{thrm:convergence-joint} can then be established.

\end{proof}

\newpage

\end{document}